\documentclass[11pt]{article}
 \usepackage[a4paper]{geometry}

\usepackage[utf8]{inputenc}   
\usepackage[T1]{fontenc}      
\usepackage{amssymb,mathrsfs}
\usepackage{amsmath,amsthm} 
\usepackage{amsfonts}
\usepackage{upgreek}
\usepackage{nicefrac}
\usepackage{authblk}
\usepackage{graphicx}
\usepackage{hyperref}
\hypersetup{
  colorlinks = true,
  citecolor = blue,
}
\usepackage{color}
\usepackage{stmaryrd}
\usepackage{enumitem}
\usepackage{url}
\usepackage{graphicx} 
\usepackage{lmodern} 
\usepackage{xcolor}
\usepackage{bbm}
\usepackage{ifthen}
\usepackage{xargs}
\usepackage[disable]{todonotes}
\usepackage{aliascnt}
\usepackage{cleveref}
\usepackage{autonum}

\newtheorem{theorem}{Theorem}
\crefname{theorem}{theorem}{Theorems}
\Crefname{Theorem}{Theorem}{Theorems}
\newtheorem*{lemma_nonumber*}{Lemma}
\newaliascnt{lemma}{theorem}
\newtheorem{lemma}[lemma]{Lemma}
\aliascntresetthe{lemma}
\crefname{lemma}{lemma}{lemmas}
\Crefname{Lemma}{Lemma}{Lemmas}
\newaliascnt{corollary}{theorem}

\aliascntresetthe{corollary}
\crefname{corollary}{corollary}{corollaries}
\Crefname{Corollary}{Corollary}{Corollaries}
\newaliascnt{proposition}{theorem}
\newtheorem{proposition}[proposition]{Proposition}
\aliascntresetthe{proposition}
\crefname{proposition}{proposition}{propositions}
\Crefname{Proposition}{Proposition}{Propositions}
\newaliascnt{remark}{theorem}

\aliascntresetthe{remark}
\crefname{remark}{remark}{remarks}
\Crefname{Remark}{Remark}{Remarks}

\crefformat{assumption}{{\textbf{A}}#2#1#3}

\newtheorem{assumptionH}{\textbf{H}\hspace{-3pt}}
\crefformat{assumptionH}{{\textbf{H}}#2#1#3}

\crefformat{assumptionD}{{\textbf{D}}#2#1#3}

\newcommand{\varphibf}{\boldsymbol{\varphi}}

\def\bareta{\bar{\eta}}
\def\tildem{\tilde{m}}

\def\Rrm{K}

\def\raymala{K}

\def\mcbb{\mathcal{B}}  

\def\msa{\mathsf{A}}

\def\msi{\mathsf{I}}

\def\rmD{\mathrm{D}}

\def\mrC{\mathrm{C}}
\def\rmC{\mrC}

\newcommandx{\invpihat}[1][1=n]{\hat{\invpi}_{#1}}
\newcommandx{\invpihattrain}[1][1=m]{\tilde{\invpi}_{#1}}

\def\RKer{R}

\def\invpi{\pi}

\def\step{\gamma}
\newcommandx{\varinf}[1][1=]{\ifthenelse{\equal{#1}{}}{\sigma^2_\infty}{\sigma^2_{\infty,#1}}}

\def\DD{\operatorname{D}}

\def\generatorA{\mathscr{E}}

\newcommandx{\sPoif}[1][1=f]{\hat{#1}}

\def\Rker{R}
\def\Qker{Q}
\def\Rkerg{R_\step}
\def\Rmalag{R_\step}
\def\Rmala{R}

\def\Qmala{R}
\def\Rula{Q}
\def\Rulag{Q_\gamma}
\def\Qula{Q}

\def\Qgam{\Qker_\step}
\def\alphamala{\tau}

\newcommandx{\genrula}[1][1=]{\ifthenelse{\equal{#1}{}}
{\generatorA_{\step}^{\scriptscriptstyle{\operatorname{ULA}}}}
{\generatorA_{#1}^{\scriptscriptstyle{\operatorname{ULA}}}}}
\newcommandx{\genrmala}[1][1=]{\ifthenelse{\equal{#1}{}}
{\generatorA_{\step}^{\scriptscriptstyle{\operatorname{MALA}}}}
{\generatorA_{#1}^{\scriptscriptstyle{\operatorname{MALA}}}}}
\newcommandx{\tgenrmala}[1][1=]{\ifthenelse{\equal{#1}{}}
{\tilde{\generatorA}_{\step}^{\scriptscriptstyle{\operatorname{MALA}}}}
{\tilde{\generatorA}_{#1}^{\scriptscriptstyle{\operatorname{MALA}}}}}

\newcommandx{\genrrwm}[1][1=]{\ifthenelse{\equal{#1}{}}
{\generatorA_{\step}^{\scriptscriptstyle{\operatorname{RWM}}}}
{\generatorA_{#1}^{\scriptscriptstyle{\operatorname{RWM}}}}}
\newcommandx{\genrbar}[1][1=]{\ifthenelse{\equal{#1}{}}
{\generatorA_{\step}^{\scriptscriptstyle{\operatorname{B}}}}
{\generatorA_{#1}^{\scriptscriptstyle{\operatorname{B}}}}}

\def\pU{U}

\def\lV{V}
\def\cFL{b}
\def\lambdaFL{\lambda}

\newcommand{\vertiii}[1]{{\left\vert\kern-0.25ex\left\vert\kern-0.25ex\left\vert #1
    \right\vert\kern-0.25ex\right\vert\kern-0.25ex\right\vert}}

\def\param{\theta}

\newcommandx{\sigS}[1][1=f]{ \hat{\sigma}^2_{N,n}(#1)}
\newcommand{\logl}[1]{%
    \IfEqCase{#1}{%
        {l}{\ell_{\operatorname{log}}}%
        {p}{\ell_{\operatorname{pro}}}%
    }[\PackageError{logl}{Undefined option to logl: #1}{}]%
}%
\newcommand{\Ub}[1]{%
    \IfEqCase{#1}{%
        {l}{\pU_{\operatorname{log}}}%
        {p}{\pU_{\operatorname{pro}}}%
    }[\PackageError{Ub}{Undefined option to Ub: #1}{}]%
}%
\newcommand{\pib}[1]{%
    \IfEqCase{#1}{%
        {l}{\invpi_{\operatorname{log}}}%
        {p}{\invpi_{\operatorname{pro}}}%
    }[\PackageError{Ub}{Undefined option to Ub: #1}{}]%
}%

\newcommandx{\rayrwm}[1][1=]{\ifthenelse{\equal{#1}{}}{K_{\step}}{K_{#1}}}

\newcommand{\ball}[2]{\mathrm{B}(#1,#2)}
\newcommand{\cball}[2]{\overline{\mathrm{B}}(#1,#2)}

\newcommandx{\flecheLimiteLoi}[1][1=\mu]{\overset{\PP_{#1}-\text{weakly}}{\underset{n\to+\infty}{\Longrightarrow}}}
\newcommandx{\flecheLimiteLoiNu}{\overset{\text{weakly}}{\underset{n\to+\infty}{\Longrightarrow}}}

\newcommandx{\flecheLimiteLoit}[1][1=x]{\overset{\PP_{#1}-\text{weakly}}{\underset{t\to+\infty}{\Longrightarrow}}}

\def\borelSet{\mathcal{B}}

\newcommandx{\functionspace}[2][1=+]{\mathsf{F}_{#1}(#2)}

\newcommandx{\VarDeux}[3][3=]{\operatorname{Var}^{#3}_{#1}\left[#2 \right]}

\newcommand{\1}{\mathbbm{1}}

\newcommand{\LeftEqNo}{\let\veqno\@@leqno}

\newcommand{\floor}[1]{\left\lfloor #1 \right\rfloor}
\newcommand{\ceil}[1]{\left\lceil #1 \right\rceil}

\newcommand{\N}{\ensuremath{\mathbb{N}}}

\newcommand{\PP}{\mathbb{P}}

\newcommand{\abs}[1]{\left\vert #1 \right\vert}

\newcommand{\tvnorm}[1]{\| #1 \|_{\mathrm{TV}}}

\newcommandx{\Vnorm}[2][1=V]{\| #2 \|_{#1}}
\newcommandx{\VnormEq}[2][1=V]{\left\| #2 \right\|_{#1}}
\newcommandx{\VnormEqs}[2][1=V]{\| #2 \|_{#1}}

\newcommandx{\estparam}[2][2=n,1=\lambda]{\widehat{\param}_{#1,#2}}
\newcommandx{\estparaml}[1][1=n]{\widetilde{\param}_{#1}}

\newcommandx{\norm}[2][1=]{\ifthenelse{\equal{#1}{}}{\Vert #2 \Vert}{\Vert #2 \Vert^{#1}}}
\newcommandx{\normEq}[2][1=]{\ifthenelse{\equal{#1}{}}{\left\Vert #2 \right\Vert}{\left\Vert #2 \right \Vert^{#1}}}
\newcommandx{\normop}[2][1=]{\Vert #2 \Vert^{#1}_{\operatorname{op}}}
\newcommandx{\normopEq}[2][1=]{\left\Vert #2 \right\Vert^{#1}_{\operatorname{op}}}
\newcommandx{\norfro}[2][1=]{\Vert #2 \Vert^{#1}_{\operatorname{F}}}

\newcommandx{\normLigne}[2][1=]{\ifthenelse{\equal{#1}{}}{\Vert #2 \Vert}{\Vert #2\Vert^{#1}}}

\newcommand{\parenthese}[1]{\left(#1 \right)}
\newcommand{\parentheseLigne}[1]{(#1 )}
\newcommand{\parentheseDeux}[1]{\left[ #1 \right]}

\newcommand{\defEns}[1]{\left\lbrace #1 \right\rbrace }
\newcommand{\defEnsLigne}[1]{\lbrace #1 \rbrace }

\newcommand{\ps}[2]{\langle#1,#2 \rangle}
\newcommand{\psEq}[2]{\left\langle#1,#2 \right\rangle}

\newcommandx\probaMarkovTilde[2][2=]
{\ifthenelse{\equal{#2}{}}{{\widetilde{\mathbb{P}}_{#1}}}{\widetilde{\mathbb{P}}_{#1}\left[ #2\right]}}

\newcommand{\bigO}{\ensuremath{\mathcal O}}

\newcommand{\plusinfty}{+\infty}

\newcounter{hypoconbis}
\newcounter{saveconbis}
\newcommand\debutH{\begin{list}
{\textbf{H\arabic{hypoconbis}}}{\usecounter{hypoconbis}}\setcounter{hypoconbis}{\value{saveconbis}}}
\newcommand\finH{\end{list}\setcounter{saveconbis}{\value{hypoconbis}}}

\def\ie{\textit{i.e.}}
\def\eqsp{\;}
\newcommand{\coint}[1]{\left[#1\right)}
\newcommand{\ocint}[1]{\left(#1\right]}
\newcommand{\ooint}[1]{\left(#1\right)}
\newcommand{\ccint}[1]{\left[#1\right]}

\newcommand{\ocintLigne}[1]{(#1]}

\newcommandx{\weight}[2][2=n]{\omega_{#1,#2}^N}

\def\rmd{\mathrm{d}}
\newcommandx\sequenceg[3][2=,3=]
{\ifthenelse{\equal{#3}{}}{\ensuremath{( #1_{#2})}}{\ensuremath{( #1_{#2})_{ #2 \geq #3}}}}

\newcommandx\sequence[3][2=,3=]
{\ifthenelse{\equal{#3}{}}{\ensuremath{\{ #1_{#2}\}}}{\ensuremath{\{ #1_{#2} \, : \,  \eqsp #2 \in #3 \}}}}
\newcommandx\sequenceW[3][2=,3=]
{\ifthenelse{\equal{#3}{}}{\ensuremath{\{ #1 \}}}{\ensuremath{\{ #1 \, :\,\eqsp #2 \in #3 \}}}}
\newcommandx\sequenceD[3][2=,3=]
{\ifthenelse{\equal{#3}{}}{\ensuremath{\{ #1_{#2}\}}}{\ensuremath{( #1)_{ #2 \in #3} }}}

\newcommandx{\sequencen}[2][2=n\in\nset]{\ensuremath{\{ #1_n \, :\, \eqsp #2 \}}}
\newcommandx{\sequencens}[2][2=n\in\nsets]{\ensuremath{\{ #1_n \, :\, \eqsp #2 \}}}
\newcommandx{\sequenceWn}[2][2=n\in\nset]{\ensuremath{\{ #1 \, :\, \eqsp #2 \}}}
\newcommandx{\sequenceWns}[2][2=n\in\nsets]{\ensuremath{\{ #1 \, :\, \eqsp #2 \}}}
\newcommandx{\sequencek}[2][2=k\in\nset]{\ensuremath{\{ #1_k \, :\, \eqsp #2 \}}}
\newcommandx{\sequenceWk}[2][2=k\in\nset]{\ensuremath{\{ #1 \, :\, \eqsp #2 \}}}
\newcommandx{\sequenceks}[2][2=k\in\nsets]{\ensuremath{\{ #1_k \, :\, \eqsp #2 \}}}
\newcommandx{\sequenceWks}[2][2=k\in\nsets]{\ensuremath{\{ #1 \, :\, \eqsp #2 \}}}
\newcommandx\sequenceDouble[4][3=,4=]
{\ifthenelse{\equal{#3}{}}{\ensuremath{\{ (#1_{#3},#2_{#3}) \}}}{\ensuremath{\{  (#1_{#3},#2_{#3}) \, :\, \eqsp #3 \in #4 \}}}}
\newcommandx\sequenceDoubleW[4][3=,4=]
{\ifthenelse{\equal{#3}{}}{\ensuremath{\{ (#1,#2) \}}}{\ensuremath{\{  (#1,#2) \, :\, \eqsp #3 \in #4 \}}}}
\newcommandx{\sequencenDouble}[3][3=n\in\N]{\ensuremath{\{ (#1_{n},#2_{n}) \, :\, \eqsp #3 \}}}

\newcommand{\wrt}{w.r.t.}

\def\iid{i.i.d.}
\def\rme{\mathrm{e}}

\def\eg{e.g.}

\def\rset{\mathbb{R}}

\def\nset{\mathbb{N}}
\def\nsets{\mathbb{N}^*}
\def\tildem{\varpi}

\newcommandx{\CPE}[3][1=]{{\mathbb E}^{#3}_{#1}\left[#2 \right]} 
\newcommand{\CPP}[3][]
{\ifthenelse{\equal{#1}{}}{{\mathbb P}\left(\left. #2 \, \right| #3 \right)}{{\mathbb P}_{#1}\left(\left. #2 \, \right | #3 \right)}}

\newcommandx{\osc}[2][1=]{\mathrm{osc}_{#1}(#2)}

\newcommand{\chunk}[4][]%
{\ifthenelse{\equal{#1}{}}{\ensuremath{{#2}_{#3:#4}}}{\ensuremath{#2^#1}_{#3:#4}}
}

\def\bfX{\mathbf{X}}

\def\param{\theta}

\def\Phibf{\mbox{\protect\boldmath$\Phi$}}

\newcounter{rmnum}

\def\bgamma{\bar{\gamma}}

\newcommand{\ensemble}[2]{\left\{#1\,:\eqsp #2\right\}}

\newcommandx{\hControlFuncOpt}[1][1=n]{g_{#1}^{\star}}
\newcommandx{\hControlFuncOptLambda}[1][1={n,\lambda}]{g_{#1}^{\star}}

\newcommandx{\ControlFuncSet}[1][1=]{\mathcal{G}_{#1}}
\newcommandx{\ControlFuncSetH}[1][1=]{\mathcal{H}_{#1}}

\newcommandx{\pen}[1][1=n]{\operatorname{pen}_{#1}}
\newcommandx{\EmpRisk}[1][1=n]{\operatorname{R}_{#1}}

\newcommandx{\PVar}[1][1=]{\ensuremath{\operatorname{Var}_{#1}}}

\newcommandx{\PCov}[1][1=]{\ensuremath{\operatorname{Cov}_{#1}}}

\newcommandx{\dlim}[1]{\ensuremath{\stackrel{#1}{\Longrightarrow}}}

\newcommandx{\MSEd}[3][1={x,\step},3=n]{\operatorname{MSE}^{#3}_{#1}(#2)}

\def\bfB{\mathbf{B}}
\def\YrULA{Y}
\def\YrMALA{X}
\def\gaStep{\gamma}
\def\ZrGaussian{Z}

\def\Psemigroup{\mathbf{P}}

\def\rmala{r}

\def\alphaMALA{\alpha_{\gamma}}
\def\VlyapUnMALA{W}

\def\Ltt{\mathtt{L}}
\def\Lttbeta{\mathtt{L}_{\beta}}
\def\mtt{\mathtt{m}}
\def\mttbeta{\mathtt{m}_{\beta}}
\def\Mtt{\mathtt{M}}

\def\tCtt{\tilde{\mathtt{C}}}
\def\tCttbeta{\tilde{\mathtt{C}}_{\beta}}

\def\Veta{V}
\def\ART{\msa^{\mathrm{RT}}}
\def\IRT{\msi^{\mathrm{RT}}}
\def\Ktt{\mathtt{K}}
\def\Kttbeta{\mathtt{K}_{\beta}}
\def\tKtt{\tilde{\mathtt{K}}}
\def\tKttbeta{\tilde{\mathtt{K}}_{\beta}}
\def\bKttbeta{\bar{\mathtt{K}}_{\beta}}

\def\Wbeta{W}
\def\baretabeta{\bar{\eta}_{\beta}}
\def\etabeta{\eta_{\beta}}

\newcommand{\txts}[1]{\textstyle #1}
\def\bdriftula{b_{\bareta,\bgamma}^{\mathrm{U}}}
\def\tbdriftula{\tilde{b}_{\bgamma}^{\mathrm{U}}}
\def\bdriftmala{b_{\bareta,\bgamma}^{\mathrm{M}}}
\def\barbdriftmala{\bar{b}_{\bareta,\bgamma}^{\mathrm{M}}}
\def\tbdriftulatGamma{\tilde{b}_{\tGamma_{1/2}}^{\mathrm{U}}}
\def\barbdriftmalabeta{\bar{b}_{\bareta,\bgamma}^{(\beta)}}
\def\ray{K}
\def\rayula{\ray^{\mathrm{U}}}
\def\raymaladrift{\ray^{\mathrm{M}}}
\def\raymaladriftD{\ray}
\def\bGamma{\bar{\Gamma}}

\def\bdriftmalabeta{b^{(\beta)}_{\bareta,\bgamma}}
\def\raymaladriftbeta{\tilde{\ray}_{\beta}}
\def\tildembeta{\tildem^{(\beta)}}
\def\rayulabeta{\ray_{\beta}}
\def\bdriftulabeta{b_{\beta}}
\def\raymaladriftbeta{\tilde{\ray}_{\beta}}
\def\raymaladriftbetaD{\tilde{\ray}_{\beta}}

\def\varepsilonula{\varepsilon}

\def\tGamma{\tilde{\Gamma}}

\def\hatGamma{\hat{\Gamma}}
\def\barA{\bar{A}}
\def\tUpsilon{\tilde{\Upsilon}}
\def\measSet{\mathsf{M}}

\def\Vly{V}

\title{On the geometric convergence for MALA under verifiable conditions}

\author[1]{Alain Durmus}
\author[2]{\'Eric Moulines}
\affil[1]{\small{Université Paris-Saclay, ENS Paris-Saclay, CNRS, Centre Borelli, F-91190 Gif-sur-Yvette, France.}}
\affil[2]{\small{CMAP - \'Ecole polytechnique, France}}

\begin{document}
\footnotetext[1]{Email: alain.durmus@ens-paris-saclay.fr}
\footnotetext[2]{Email: eric.moulines@polytechnique.edu}

\maketitle

\begin{abstract}
  While the Metropolis Adjusted Langevin Algorithm (MALA) is a popular and widely used Markov chain Monte Carlo method, very few papers derive conditions
  that ensure its convergence. In particular, to the authors' knowledge, assumptions that are both easy to verify and guarantee geometric convergence, are still missing. In this work, we establish $V$-uniformly geometric convergence for MALA under mild assumptions about the target distribution. Unlike previous work, we only consider tail and smoothness conditions for the potential associated with the target distribution. These conditions are quite common in the MCMC literature and are easy to verify in practice. Finally, we pay special attention to the dependence of the bounds we derive on the step size of the Euler-Maruyama discretization, which corresponds to the proposal  Markov kernel of  MALA.

\end{abstract}


\section{Introduction}
This paper deals with the convergence of the Metropolis Adjusted Langevin Algorithm (MALA) for sampling from a positive target probability density $\pi$ on $(\rset^d,\mcbb(\rset^d))$, where $\mcbb(\rset^d)$ is the Borel $\sigma$-field of $\rset^d$ endowed with the Euclidean topology.
For simplicity, we also denote by $\pi$, the distribution corresponding to the density $\pi$, and let $U = -\log \pi$ be the associated potential function.
MALA is a Markov Chain Monte Carlo (MCMC) method based on the Langevin diffusion associated with $\pi$:
\begin{equation}
  \label{eq:def_langevin}
  \rmd \bfX_t =  -\nabla U(\bfX_t) \rmd t + \sqrt{2}  \rmd \bfB_t \eqsp,
\end{equation}
where $\sequenceg{\bfB}[t][0]$ is a $d$-dimensional Brownian motion. It is known that under mild conditions this diffusion admits a strong solution $\sequenceg{\bfX^{(x)}}[t][0]$ for any starting point $x\in \rset^d$ and defines a Markov semigroup $\sequenceg{\Psemigroup}[t][0]$ for any $t \geq 0$, $x \in \rset^d$ and $\msa \in \mcbb(\rset^d)$ by $\Psemigroup_t(x,\msa) = \PP (\bfX_t^{(x)} \in \msa)$. Moreover, this Markov semigroup admits $\pi$ as its unique stationary measure, is ergodic and even $V$-uniformly geometrically ergodic with additional assumptions on $U$ (see \cite{roberts:tweedie-Langevin:1996,Mattingly2002185}). However, sampling a path solution of \eqref{eq:def_langevin} is a real challenge in most cases and discretizations are used instead to obtain a Markov chain with similar long-time behaviour. Here we consider the
Euler-Maruyama discretization, associated  with
\eqref{eq:def_langevin}, defined for all $k \geq 0$ by
\begin{equation}
  \label{eq:def_ula}
  \YrULA_{k+1} = \YrULA_k - \gaStep \nabla U(\YrULA_k) + \sqrt{2\gaStep} \ZrGaussian_{k+1} \eqsp,
\end{equation}
where $\gaStep$ is the step size of the discretization and
$\sequence{\ZrGaussian}[k][\nsets]$ is an \iid~sequence of $d$-dimensional
standard Gaussian random variables. This algorithm
has been suggested  \cite{ermak:1975,parisi:1981} and later studied by \cite{grenander:1983,grenander:miller:1994,neal:1992,roberts:tweedie-Langevin:1996}. Following
\cite{roberts:tweedie-Langevin:1996}, this algorithm is called the
Unadjusted Langevin Algorithm (ULA). A drawback of this method is that even if the Markov chain $\sequence{\YrULA}[k][\nset]$ has a unique stationary distribution $\pi_{\gaStep}$ and is ergodic (which is guaranteed under mild assumptions about $U$), $\pi_{\gaStep}$ is different from $\pi$ most of the time.
To circumvent this problem, it was proposed in 
\cite{rossky:doll:friedman:1978,roberts:tweedie-Langevin:1996}  to use the Markov kernel associated with the recursion defined by the Euler-Maruyama discretization \eqref{eq:def_ula} as a proposal kernel in a Metropolis-Hastings algorithm defining a new Markov chain $\sequencek{\YrMALA}$ by:
\begin{equation}
  \label{eq:def_MALA}
  \YrMALA_{k+1}=   \YrMALA_{k} + \1_{\rset_+}(U_{k+1}-\alpha_{\gamma}( \YrMALA_{k}, \tilde{\YrULA}_{k+1}))\{\tilde{\YrULA}_{k+1}-\YrMALA_k\} \eqsp,
\end{equation}
where $\tilde{\YrULA}_{k+1} = \YrMALA_k-\gamma \nabla U(\YrMALA_k) + \sqrt{2\gamma} \ZrGaussian_{k+1}$, $\sequenceks{U}$ is a sequence of \iid~uniform random variables on $\ccint{0,1}$ and $\alpha_{\gamma} : \rset^{2d} \to \ccint{0,1}$ is the usual Metropolis acceptance ratio defined in \eqref{eq:def-alpha-MALA_0}.
This algorithm is called Metropolis Adjusted Langevin Algorithm (MALA) and has since been used in many applications. It may be surprising that a complete theoretical understanding of its properties is still lacking. While some interesting results have been derived in \cite{roberts:tweedie-Langevin:1996,bourabee:hairer:2013}, they are only partially satisfactory. More precisely, \cite{roberts:tweedie-Langevin:1996} does not give practical conditions for either the step size or the target distribution that guarantee geometric convergence for MALA. On the other hand, \cite{bourabee:hairer:2013,pmlr-v75-dwivedi18a,chewi2021optimal} provides non-quantitative or quantitative bounds between the iterates of the Markov kernel associated with MALA and $\pi$, but they do not establish convergence. In this paper, we address this problem and give practical conditions to ensure that MALA $V$-uniform is geometrically ergodic. Part of our results were already included in the pre-publication \cite{brosse2019diffusionv2}. Since they are not the main focus of this work and have attracted independent interest, we have decided to extract them from the forthcoming revision of \cite{brosse2019diffusionv2} and extend them. 

The paper is organized as follows. First we state the assumptions on the potential $U$ and our main results. In \Cref{sec:comp-with-exist} we provide a detailed comparison of our paper with the existing literature. Finally, in \Cref{subsec:geom-ergodicity-ula} we summarize the proofs of our results. Finally, we also consider a generalization of our results in \Cref{sec:extens-cref}. However, the resulting proofs are more complicated and, in our opinion, would hinder the flow of the paper. Therefore, we have moved them to another section.


\subsection*{Notation and convention}
Denote by $\mathcal{B}(\rset^d)$ the Borel $\sigma$-field of $\rset^d$ and by $\functionspace[]{\rset^d}$ the set of all Borel measurable functions on $\rset^d$ and for $f \in \functionspace[]{\rset^d}$, $\Vnorm[\infty]{f}= \sup_{x \in \rset^d} \abs{f(x)}$. Denote by $\measSet(\rset^d)$ the space of finite signed measure on $(\rset^d, \mathcal{B}(\rset^d))$ and $\measSet_0(\rset^d) =  \{\mu \in \measSet(\rset^d)\ | \ \mu(\rset^d) = 0\}$. For $\mu \in \measSet(\rset^d)$  and $f \in \functionspace[]{\rset^d}$ a $\mu$-integrable function, denote by $\mu(f)$ the integral of $f$ \wrt~$\mu$. Let $\Vly: \rset^d \to \coint{1,\infty}$ be a measurable function. For $f \in \functionspace[]{\rset^d}$, the $\Vly$-norm of $f$ is given by $\Vnorm[\Vly]{f}= \sup_{x \in \rset^d} |f(x)|/\Vly(x)$. For $\mu \in \measSet(\rset^d)$, the $\Vly$-total variation distance of $\mu$ is defined as
\begin{equation}
\Vnorm[\Vly]{\mu} = \sup_{f \in \functionspace[]{\rset^d}, \Vnorm[\Vly]{f} \leq 1}  \abs{\int_{\rset^d } f(x) \rmd \mu (x)} \eqsp
\end{equation}
If $\Vly \equiv 1$, then $\Vnorm[\Vly]{\cdot}$ is the total variation  denoted by $\tvnorm{\cdot}$.


\section{Main results}
\label{sec:geom-ergodicity-mala}

Denote by $\rmala_{\gamma}$ the proposal transition density associated  to the Euler-Maruyama discretization \eqref{eq:def_ula} with stepsize $\gamma >0$, \ie, for any $x,y \in \rset^d$,
\begin{equation}
  \label{eq:def_r_gamma_ULA}
  \rmala_{\gamma}(x,y) = (4\uppi \gamma)^{-d/2} \exp \parenthese{-(4 \gamma)^{-1}\norm[2]{y-x+ \gamma \nabla U(x)}} \eqsp.
\end{equation}
Then,  the Markov kernel $\Rmala_{\gamma}$ of the MALA algorithm \eqref{eq:def_MALA} is given for $\gaStep>0$, $x\in\rset^d$, and $\msa \in\borelSet(\rset^d)$ by
\begin{align}
\label{eq:def-kernel-MALA}
  \Rmala_{\gamma}(x ,\msa)& = \int_{\rset^d}  \1_{\msa}(y) \alphaMALA(x,y)\rmala_{\gamma}(x,y) \rmd y + \updelta_{x}(\msa) \int_{\rset^d} \{1 - \alphaMALA(x,y)\}\rmala_{\gamma}(x,y) \rmd y  \eqsp,\\
  \label{eq:def-alpha-MALA_0}
 \alphaMALA(x,y) &= 1\wedge\parentheseDeux{\frac{\pi(y)r_{\gamma}(y,x)}{\pi(x)r_{\gamma}(x,y)}} \eqsp.
\end{align}

It is well-known, see \eg~\cite{roberts:tweedie-Langevin:1996}, that
for any $\gamma >0$, $\Rmala_{\gamma}$ is reversible with respect to
$\pi$ and $\pi$-irreducible.

We establish that MALA is $V$-uniformly
geometrically ergodic, under the following assumptions on the
potential $U$.
\begin{assumptionH}
  \label{ass:regularity-U}
The function $U : \rset^d \to \rset$ is twice continuously differentiable. In addition, $\nabla U(0) = 0$ and there exists $\Ltt \geq 0$ such that $\sup_{x \in \rset^d}\norm{\rmD^2 U(x)} \leq \Ltt$.
\end{assumptionH}

The condition $\nabla U(0) = 0$ is satisfied (up to a translation) as soon as $U$ has a local minimum, which is the case  when $\lim_{\norm{x} \to \plusinfty} U(x) = \plusinfty$, since $U$ is continuous. It could be relaxed but at the cost of more complicated computations that would hinder the derivation of our proofs.

The other condition in \Cref{ass:regularity-U} is standard in the analysis of ULA. In particular, it implies that $\nabla U$ is Lipschitz, which is a necessary condition to ensure that this scheme is stable; see \eg~\cite{roberts:tweedie-Langevin:1996}.
Finally, note that \Cref{ass:regularity-U} implies that for $x\in\rset^d$, $\norm{\nabla U(x)} \leq \Ltt \norm{x}$.

\begin{assumptionH}
  \label{ass:regularity-U_C3}
The function $U : \rset^d \to \rset$ is three times continuously differentiable. In addition, there exists $\Mtt \geq 0$ such that  $\sup_{x \in \rset^d}\norm{\rmD^3 U(x)} \leq \Mtt$.
\end{assumptionH}

We consider the two conditions \Cref{ass:regularity-U} and \Cref{ass:regularity-U_C3} separately. In fact, we derive a non-quantitative convergence result under \Cref{ass:regularity-U} only, while \Cref{ass:regularity-U_C3} allows us to obtain quantitative convergence bounds. While \Cref{ass:regularity-U} and \Cref{ass:regularity-U_C3} impose regularity constraints on the potential $U$, we now consider tail conditions.

\begin{assumptionH}
  \label{ass:curvature_U}
  There exist $\mtt >0$ and $\Ktt \geq 0$ such that for any $x,y \in \rset^d$, $\norm{x} \geq \Ktt$ and $\norm{y} = 1$,
  \begin{equation}
    \label{eq:7}
\rmD^2U(x) \{y\}^{\otimes 2}  \geq  \mtt  \eqsp.
  \end{equation}
\end{assumptionH}
Note that under \Cref{ass:regularity-U} and \Cref{ass:curvature_U}, for any $x,y \in \rset^d$, $\norm{y} = 1$, it holds that
  \begin{equation}
    \label{eq:7_2}
\rmD^2U(x) \{y\}^{\otimes 2}  \geq  \mtt -(\mtt + \Ltt) \1_{\ball{0}{\Ktt}}(x) \eqsp.
  \end{equation}
In the case $\Ktt = 0$, \Cref{ass:curvature_U} boils down requiring that $U$ is strongly convex with convexity constant equals to  $\mtt$. However, when $\Ktt >0$, \Cref{ass:curvature_U} is a slight strengthening of the strong convexity at infinity condition considered in \cite{chen1997estimation,eberle:2015}: there exist $\mtt' >0$ and $\Ktt' \geq 0$ such that for any $x,y \in\rset^d$, $\norm{x-y} \geq \Ktt'$
\begin{equation}
  \label{eq:strong_convex_infi}
\text{    $\ps{\nabla U(x) - \nabla U(y)}{x-y} \geq \mtt' \norm{x-y}^2$  } \eqsp.
  \end{equation}
  Indeed, if \eqref{eq:strong_convex_infi} holds for any $x,y\in\rset^d$ satisfying $\norm{x} \vee \norm{y} \geq \Ktt'$ in place of $\norm{x-y} \geq \Ktt'$, then an easy computation implies that \Cref{ass:curvature_U} holds with $\mtt \leftarrow \mtt'$ and $\Ktt \leftarrow \Ktt'+1$. Besides, \Cref{lem:equiv_str_conv_inf_hessian} shows that the converse is true. 
Finally, while the condition \eqref{eq:strong_convex_infi} for $x,y \in\rset^d$, $\norm{x-y} \geq \Ktt'$, is weaker than  \Cref{ass:curvature_U},   it can be more convenient in many situations to verify that the latter holds.
Note that under \Cref{ass:regularity-U} and \Cref{ass:curvature_U}, $\mtt \leq \Ltt$. In addition, we show in \Cref{lem:quadratic_behaviour} that under these two conditions,
there exists $\tKtt \geq 0$ such that for any $x \not \in\ball{0}{\tKtt}$, $\ps{\nabla U(x)}{x} \geq (\mtt/2) \norm[2]{x}$. Therefore,
 for any $x \not \in \ball{0}{\tKtt}$, we have by the Cauchy-Schwarz inequality that $\norm{\nabla U(x)} \geq \mtt \norm{x}$.

We can also consider the following generalization for $\beta \in \coint{0,1}$,
\begin{assumptionH}[$\beta$]
  \label{ass:curvature_U_alpha}
  There exist $\mttbeta >0$ and $\Lttbeta,\Kttbeta \geq 0$ such that for any $x,y \in \rset^d$, $\norm{x} \geq \Kttbeta$ and $\norm{y} = 1$,
    \begin{equation}
    \label{eq:7_alpha}
\Lttbeta/(1+\norm{x}^{3\beta/4}) \geq     \rmD^2U(x) \{y^{\otimes 2}\} \geq \mttbeta/(1+\norm{x}^{\beta})  \eqsp.
\end{equation}
\end{assumptionH}
Note that under \Cref{ass:regularity-U} and \Cref{ass:curvature_U_alpha}$(\beta)$ it holds that $\mttbeta \leq \Lttbeta$ and
for any $x,y \in \rset^d$,  $\norm{y} = 1$,
  \begin{equation}
    \label{eq:7_alpha_2}
    \begin{aligned}
    \rmD^2U(x) \{y^{\otimes 2}\}& \geq \mttbeta/(1+\norm{x}^{\beta})  - (\mttbeta/(1+\norm{x}^{\beta}) + \Ltt)  \1_{\ball{0}{\Kttbeta}}(x) \eqsp,\\
\rmD^2U(x) \{y^{\otimes 2}\}& \leq \Lttbeta/(1+\norm{x}^{3\beta/4})  + (-\Lttbeta/(1+\norm{x}^{3\beta/4}) + \Ltt)  \1_{\ball{0}{\Kttbeta}}(x)
\eqsp.
\end{aligned}
\end{equation}

We show in \Cref{lem:quadratic_behaviour_alpha} that under  \Cref{ass:regularity-U} and \Cref{ass:curvature_U_alpha}$(\beta)$,
there exists $\tKttbeta \geq 0$ such that for any $x \not \in\ball{0}{\tKttbeta}$, $\ps{\nabla U(x)}{x} \geq (\mttbeta/2) \norm[2]{x}/(1+\norm[\beta]{x})$. Therefore,
 for any $x \not \in \ball{0}{\tKttbeta}$, we have by the Cauchy-Schwarz inequality that $\norm{\nabla U(x)} \geq \mttbeta \norm{x}/(1+\norm[\beta]{x})$. Finally, \Cref{lem:upper_bound_behaviour_alpha} shows that under the same conditions, there exists $\bKttbeta \geq 0$ such that $\norm{\nabla U(x)} \leq 2\Lttbeta\norm{x}/(1+\norm[3\beta/4]{x})$ for $x \in\rset^d$, $\norm{x} \geq \bKttbeta$.

We first present non-quantitative $V$-uniformly geometric ergodicity results under \Cref{ass:regularity-U}, \Cref{ass:curvature_U}; then we present a quantitative convergence statement under the additional assumption \Cref{ass:regularity-U_C3}. We postpone the statement of the results under \Cref{ass:curvature_U_alpha} (instead of \Cref{ass:curvature_U}) to \Cref{sec:extens-cref}.

Define for any $\eta >0$, $V_{\eta}: \rset^d \to \coint{1,\plusinfty}$ for any $x \in \rset^d$ by
\begin{equation}
  \label{eq:def_V_eta}
  \Veta_{\eta}(x) = \exp(\eta \norm[2]{x}) \eqsp.
\end{equation}

\begin{theorem}
  \label{theo:V-geom_ergo_MALA_just_C2}
  Assume \Cref{ass:regularity-U} and \Cref{ass:curvature_U}. Then, there exists $\Gamma >0$ (defined in \eqref{eq:def_Gamma}) such that for any $\gamma \in \ocint{0,\Gamma}$, there exist $C_{\gamma} \geq 0$ and $\rho_{\gamma}\in\coint{0,1}$ such that for any $x \in \rset^d$,
  \begin{equation}
    \label{eq:5}
    \Vnorm[\Veta_{\bareta}]{\updelta_x \Rmala_{\gamma} - \pi} \leq C_{\gamma} \rho_{\gamma}^{k} \Veta_{\bareta}(x) \eqsp,
  \end{equation}
  where $\bareta = \mtt/16$.
\end{theorem}
\begin{proof}
  The proof  is postponed to \Cref{sec:proof:theo:V-geom_ergo_MALA_just_C2}.
\end{proof}

The constants $C_{\gamma},\rho_{\gamma}$ appearing in \Cref{theo:V-geom_ergo_MALA_just_C2} are non quantitative. Indeed, the proof of \Cref{theo:V-geom_ergo_MALA_just_C2} only relies on a Foster-Lyapunov drift condition and the fact that all compact sets are $1$-small.

To obtain geometric convergence with quantitative constants we need to consider the additional regularity assumption \Cref{ass:regularity-U_C3}.

\begin{theorem}
  \label{theo:V-geom_ergo_MALA}
  Assume \Cref{ass:regularity-U}, \Cref{ass:regularity-U_C3} and \Cref{ass:curvature_U}. Then, there exist $\bGamma,\barA_{\bareta} > 0$ (defined in \eqref{eq:def_bGamma_final} and \eqref{eq:def_bA_final}), such that for any $\bgamma \in\ocint{0,\bGamma}$, there exist $C_{\bgamma} \geq 0$ and $\rho_{\bgamma}\in\coint{0,1}$ (given in \eqref{eq:cst_bornes_mala}) satisfying for any $x \in \rset^d$ and $\gamma \in\ocint{0,\bgamma}$,
  \begin{equation}
    \label{eq:V-geom_ergo_MALA}
    \Vnorm[\Veta_{\bareta}]{\updelta_x \Rmala_{\gamma} - \pi} \leq C_{\bgamma} \rho_{\bgamma}^{\gamma k} \{ \Veta_{\bareta}(x) + \pi(\Veta_{\bareta})\} \eqsp,
  \end{equation}
  where $    \bareta = \mtt/ 16$ and $\pi(\Veta_{\bareta}) \leq \barA_{\bareta}$.
\end{theorem}
\begin{proof}
  The proof is postponed to \Cref{sec:proof:theo:V-geom_ergo_MALA}.
\end{proof}
Note that the constants $C_{\bgamma},\rho_{\bgamma}$ only depend on the characteristics of $U$ appearing in the conditions \Cref{ass:regularity-U}, \Cref{ass:regularity-U_C3} and \Cref{ass:curvature_U} and are independent of the stepsize $\gamma$. As a result, \Cref{theo:V-geom_ergo_MALA} establishes that the rate of convergence of MALA, \ie, $\gamma \log(\rho_{\bgamma})$ scales linearly with respect to the stepsize $\gamma$. It is in accordance with the convergence rates of ULA \cite{eberle2018quantitative,debortoli2018back} and with the result that under appropriate conditions,  MALA is at the first order an approximation of the Langevin diffusion. Indeed,  \cite[Lemma 7]{brosse2019diffusionv2} (see also \cite{refId0,lelievre_stoltz_2016,roberts:rosenthal:1998} and the references therein for similar results)  shows that for any $\varphi : \rset^d \to \rset$, $\rmC^{\infty}$ with compact support, $\Qmala_{\gamma} \varphi(x) - \Psemigroup_{\gamma} \varphi(x) = \gamma^{2}(1+\norm[q]{x})  \psi_{\gamma,\varphi}(x)$, for $q \geq 0$ and some function $\psi_{\gamma,\varphi} : \rset^d \to \rset$ satisfying $\sup_{x \in \rset^d, \, \gamma \in\ocint{0,\bgamma}} \abs{\psi_{\gamma,\varphi}(x)} < \plusinfty$. On the other hand, if $(\Psemigroup_t)_{t \geq 0}$ is $V_{\mathrm{L}}$-uniformly geometrically ergodic, \ie, there exist $V_{\mathrm{L}} : \rset^d \to \coint{0,1}$, $\rho_{\mathrm{L}} \in \coint{0,1}$ and $C_{\mathrm{L}} \geq 0$ such that for any $t \geq 0$ and $x \in\rset^d$, $\Vnorm[V_{\mathrm{L}}]{\updelta_x \Psemigroup_t - \pi} \leq C_{\mathrm{L}} \rho_{\mathrm{L}}^{t}$,  $\Psemigroup_{\gamma}$ as a discrete Markov kernel converges to $\pi$ with a convergence rate, \ie, $\gamma\log(\rho_{\mathrm{L}})$ which scales linearly with respect to $\gamma$. Hence, such a convergence is expected for $\Qmala_{\gamma}$ with some constants $C_{\bgamma}$ and $\rho_{\bgamma}$ independent of $\gamma \in \ocint{0,\bgamma}$ as stated in \Cref{theo:V-geom_ergo_MALA}. This type of convergence is important to obtain bounds on Poisson solution associated with $\Qmala_{\gamma}$ which holds uniformly with respect to the discretization parameters. We refer to \cite{brosse2019diffusionv2,lelievre_stoltz_2016,durmus:enfroy:2021} for further discussions on this matter.


The proof of \Cref{theo:V-geom_ergo_MALA} consists in establishing explicit minorization and drift conditions for $\Rmala_{\gamma}$ for $\gamma \in \ocint{0,\bgamma}$ for some $\bgamma >0$; see \eg~\cite[Chapter~19]{douc:moulines:priouret:soulier:2018}. In particular,  to obtain the stated dependence with respect to the stepsize $\gamma$, we show that for some $\bgamma >0$:
\begin{enumerate}[label=(\Roman*)]
\item\label{item:condition_ergo_I} there exist $\lambdaFL\in\ooint{0,1}$ and $\cFL<\plusinfty$ such that for all $\step\in\ocint{0,\bgamma}$
  \begin{equation}
    \label{eq:def-discrete-drift}
\RKer_\step \lV_{\bareta} \leq \lambdaFL^\step \lV_{\bareta} + \step \cFL  \eqsp;
\end{equation}
\item\label{item:condition_ergo_II}  there exists $\varepsilon\in\ocint{0,1}$ such that for all $\step\in\ocint{0,\bgamma}$ and $x,x'\in\defEnsLigne{\lV_{\bareta} \leq M}$,
\begin{equation}\label{eq:def-minorization}
  \tvnorm{\Rker_\step^{\ceil{1/\step}}(x,\cdot) - \Rker_\step^{\ceil{1/\step}}(x',\cdot)} \leq 2(1 - \varepsilon) \eqsp,
\end{equation}
where
\begin{equation}
  M>\parenthese{\frac{4b \lambda^{-\bgamma}}{(1-\lambda)\log(1/\lambda)}} \vee 1 \eqsp.
\end{equation}
\end{enumerate}
Then,  \ref{item:condition_ergo_I} implies by an easy induction that for any $\gamma \in \ocint{0,\bgamma}$,
\begin{equation}\label{eq:def-discrete-drift_2}
  \RKer_\step^{\ceil{1/\gamma}} \lV_{\bareta} \leq \lambdaFL \lV_{\bareta} +  \cFL(1+\bgamma)  \eqsp.
\end{equation}
Therefore, applying
\cite[Theorem~19.4.1]{douc:moulines:priouret:soulier:2018} to $\RKer^{\ceil{1/\gamma}}_{\gamma}$ for $\gamma \in \ocint{0,\bgamma}$ using \ref{item:condition_ergo_II} and \eqref{eq:def-discrete-drift_2}, it follows  that \Cref{theo:V-geom_ergo_MALA} holds  and $\pi(V_{\bareta}) < \plusinfty$. Accordingly, it is enough to show that conditions \ref{item:condition_ergo_I} and \ref{item:condition_ergo_II} hold.

\subsection{Comparison with existing litterature}
\label{sec:comp-with-exist}

MALA has been shown to be uniformly geometrically ergodic in \cite{roberts:tweedie-Langevin:1996} but under very restrictive conditions which we recall.
Define for any $x \in\rset^d$,
\begin{equation}
  \ART(x) = \ensemble{y \in\rset^d}{\pi(x)r_{\gamma}(x,y)\leq \pi(y)r_{\gamma}(y,x)} \eqsp, \quad \IRT(x) = \ensemble{y\in\rset^d}{\norm{y} \leq \norm{x}} \eqsp.
\end{equation}
We say that $\ART$ converges inward in $r_{\gamma}$ for $\gamma >0$ if
\begin{equation}
  \label{eq:1}
  \lim_{\norm{x} \to \plusinfty}\int_{\rset^d} \1_{\ART(x)\cap \IRT(x)}(y)r_{\gamma}(x,y) \rmd x = 0 \eqsp.
\end{equation}

Define for any $a \geq 0$, $  \VlyapUnMALA_a(x) = \exp(a \norm{x})$.

\begin{theorem}[\protect{\cite[Theorem 4.1]{roberts:tweedie-Langevin:1996}}]
  \label{theo:2}
Let $\gamma >0$.  Assume that there exists $\eta >0$ such that
\begin{equation}
  \label{eq:tail_U}
    \liminf_{\norm{x} \to \plusinfty} \{\norm{x} - \normLigne{x - \gamma \nabla U(x)} \}\geq \eta \eqsp,
  \end{equation}
  and $\ART$ converges inward in $r_{\gamma}$. Then $\Rmala_{\gamma}$ is $\VlyapUnMALA_{a}$-uniformly geometrically ergodic for $a \in \ooint{0,\gamma \eta}$, \ie, there exist $C_{\gamma} \geq 0$ and $\kappa_{\gamma} >0$ such that
  for any $x \in \rset^d$,
  \begin{equation}
    \label{eq:2_0}
    \Vnorm[\VlyapUnMALA_{a}]{\updelta_x \Qmala_{\gamma} - \pi} \leq C_{\gamma} \rme^{-\kappa_{\gamma} \gamma k} \VlyapUnMALA_{a}(x) \eqsp.
  \end{equation}
\end{theorem}

Let us comment on \Cref{theo:2}. First, while \eqref{eq:tail_U} is relatively easy to verify under mild assumptions about the tail of $\pi$, deriving practical conditions for $U$ that ensure that $\ART$ converges inward into $r_{\gamma}$ is very difficult, and \cite{roberts:tweedie-Langevin:1996}
only succeeds in showing that this holds in the one-dimensional setting and under strong conditions for $U$. Second, \Cref{theo:2} is not quantitative and therefore the constants in the geometric convergence of $Q_{\gamma}$ to $\pi$ may depend strongly on the step size $\gamma$, which in general must be chosen small to ensure that a non-negligible fraction of candidates is accepted during the algorithm. For example, from the optimal scaling for MALA in \cite{roberts:rosenthal:1998}, the step size should scale as $d^{-1/3}$ with the dimension $d \to \plusinfty$ in the ideal scenario $\pi = \pi_1^{\otimes d}$ for a one-dimensional distribution $\pi_1$.
Even if these constants were independent of $\gamma$\footnote{which is the case for the one-dimensional distribution considered in \cite{roberts:tweedie-Langevin:1996} after a careful review of the computations}, we can see that the parameter $a$ of the Lyapunov function $\VlyapUnMALA_{a}$ must be chosen proportional to the step size $\gamma$, which implies that convergence behaves poorly with respect to $\gamma$ for unbounded functions. Finally, this result does not give a recommendation for the choice of $\gamma$ that ensures geometric convergence. Note that our results \Cref{theo:V-geom_ergo_MALA_just_C2} and \Cref{theo:V-geom_ergo_MALA} address all these issues.

Recent studies \cite{pmlr-v75-dwivedi18a,chewi2021optimal} based on conductance arguments \cite{lovasz:vempala:2007,kannan:lovasz:simonovits:1995} establish quantitative complexity bounds for MALA in the case where the potential $U$ is strongly convex. More precisely, given a precision $\epsilon > 0$, these works are interested in finding a minimal number of iterations $N_{\epsilon} \in \nsets$ and a step size $\gamma_{\epsilon} > 0$ that ensures that $\tvnorm{\xi \Rmala_{\gamma_{\epsilon}}^{N_{\epsilon}}-\pi}\leq \epsilon$, where $\xi$ is either a warm start or a well-chosen initial distribution. In contrast to these works, we do not impose a strong convexity condition and our result can be applied to any initial distribution. Finally, we show $V$-uniform geometric ergodicity, which is a  stronger convergence guarantee.

Finally, we mention \cite{bourabee:hairer:2013}, which studies the case where $U$ satisfies \Cref{ass:curvature_U} but potentially violates \Cref{ass:regularity-U}, \ie, $x \mapsto \norm{\rmD^2 U(x)}$ can be unbounded.  \cite[Theorem 3.1]{bourabee:hairer:2013} shows that under suitable regularity conditions and \Cref{ass:curvature_U}, there exist $\rho \in \coint{0,1}$, $\bgamma > 0$ and $C_2 \geq 0$ such that for any $E_0 \in \rset$, there exists $C_1(E_0) \geq 0$ such that for any $x \in \rset^d$, $U(x) \leq E_0$ and $\gamma \in\ocint{0,\bgamma}$,
\begin{equation}
  \label{eq:6}
  \tvnorm{\updelta_x \Rmala_{\gamma}^{k\floor{1/\gamma}} - \pi} \leq C_1(E_0)\{\rho^{k} + \rme^{-C_2/\gamma^{1/4}}\} \eqsp.
\end{equation}
We show that in the case $U$ is Lipschitz the extra term $\rme^{-C_2/\gamma^{1/4}}$ can be omitted and the convergence occurs in a particular $V$-norm.



\section{Proof of the main results}
\label{subsec:geom-ergodicity-ula}

\subsection{Bounds on the acceptance ratio}
The analysis of MALA is naturally  related to the study of the ULA algorithm. More precisely, since for any $x \in\rset^d$ and $\msa \in \mcbb(\rset^d)$, the Markov kernel corresponding to ULA \eqref{eq:def_ula} is given by
\begin{equation}
  \label{eq:def_kernel_ula}
  \Rulag(x,\msa) = \int_{\rset^d}  \1_{\msa}(x- \gamma \nabla U(x) + \sqrt{2\gamma} z) \min(1,\rme^{-\alphamala_{\gamma}(x,z)})  \varphibf(z) \rmd z
\end{equation}
 where $\varphibf(z) = (2\uppi)^{-d/2} \rme^{-\norm{z}^2/2}$ is the density of the $d$-dimensional standard Gaussian distribution and
\begin{align}
\label{eq:def-kernel-MALA}
  &
\textstyle{\Rmalag(x ,\msa) = \int_{\rset^d}  \1_{\msa}(x- \gamma \nabla U(x) + \sqrt{2\gamma} z) \min(1,\rme^{-\alphamala_{\gamma}(x,z)})  \varphibf(z) \rmd z}  \\
  \nonumber
&\phantom{--}+
\textstyle{\updelta_x(\msa) \int_{\rset^d} \{1- \min(1,\rme^{-\alphamala_{\gamma}(x,z)})} \} \varphibf(z) \rmd z \eqsp,\\
    \label{eq:def-alpha-MALA}
& \textstyle{\alphamala_{\gamma}(x,z) = \pU(x - \gamma \nabla U(x) + \sqrt{2 \gamma}z) - \pU(x)} \\
  \nonumber
&  \phantom{--}
\textstyle{+ (1/2) \{\norm[2]{z-(\gamma/2)^{1/2}\defEns{\nabla U(x) + \nabla U(x-\gamma \nabla U(x) + \sqrt{2\gamma} z)}} -\norm[2]{z} \}}\eqsp.
\end{align}
 the difference between the two Markov kernels can be expressed for any bounded measurable function $f:\rset^d\to\rset$ by
\begin{multline}
\label{eq:diff-rula-rmala}
  \Rmalag f(x)  - \Rulag f(x) = \int_{\rset^d}\{f(x) - f(x-\gamma \nabla U(x) + \sqrt{2 \gamma} z)\} \\
  \times \{1 - \min(1,\rme^{-\alphamala_{\gamma}(x,z)}) \} \varphibf(z) \rmd z \eqsp.
\end{multline}
Since $1-\min(1,\rme^{-t}) \leq \max(0,t)$ for any $t \in \rset$,
properties of ULA can then be transferred to MALA from perturbation
arguments achieved by a careful analysis of $\alphamala_{\gamma}$.
Most analyses of MALA dealing with either its convergence
\cite{refId0,eberle:2014} or its optimal scaling
\cite{roberts:rosenthal:1998} establish that
$\alphamala_{\gamma}(x,z)$ is of order $\bigO(\gamma^{3/2})$ for fixed $x,z\in\rset^d$. More
precisely, we have:

\begin{lemma}
  \label{lem:bound_alpha_mala_1}
Assume \Cref{ass:regularity-U} and \Cref{ass:regularity-U_C3} holds.
 Then, for any $\bgamma >0$, there exists an explicit constant (see \eqref{eq:def_C_1_bgamma}) $C_{1,\bgamma} < \infty$ such that for any $x,z \in \rset^d$, $\gamma \in \ocint{0,\bgamma}$, it holds
  \begin{equation}
          \label{eq:lem:bound_alpha_mala_1}
\abs{      \alphamala_{\gamma}(x,z) }\leq
C_{1,\bgamma} \gamma^{3/2}\{\norm[2]{z}+\norm[4]{z} + \norm[2]{x}\}  \eqsp.
  \end{equation}
\end{lemma}

This result is a first step in the proof of \Cref{theo:V-geom_ergo_MALA}, from which we will be able to transfer the explicit minorization condition of ULA (\Cref{propo:small_set_ula}) to MALA; see \Cref{lem:diff_tv_MALA_ULA} and \Cref{propo:small_set_mala} below. Unfortunately, the dependence on $\norm{x}$ in the upper bound of \eqref{eq:lem:bound_alpha_mala_1} prevents us from doing the same for the Lyapunov drift condition of ULA \Cref{propo:super_lyap_ula}. Instead, we rely on the following upper bound, which does not have the correct order with respect to $\gamma$, but is independent of $x$.

\begin{lemma}
  \label{lem:bound_alpha_mala_2}
  Assume \Cref{ass:regularity-U} and \Cref{ass:curvature_U}. Then, for any $\bgamma \in \ocint{0, \mtt^3/(4\Ltt^4)}$, there exists an explicit constant (see \eqref{eq:def_C_2_bgamma}) $C_{2,\bgamma} < \infty$ such that for any  $\gamma \in \ocint{0 ,\bgamma}$, $x,z \in\rset^d$, $\norm{x} \geq \max(2 \Ktt, \tKtt)$, $\tKtt$ given in \Cref{lem:quadratic_behaviour}, and $\norm{z} \leq \norm{x}/(4\sqrt{2\gamma})$, it holds
  \begin{equation}
          \label{eq:lem:bound_alpha_mala_2}
    \alphamala_{\gamma}(x,z) \leq       C_{2,\bgamma}\gamma\norm[2]{z} \eqsp.
  \end{equation}
\end{lemma}
To show \Cref{lem:bound_alpha_mala_1} and \Cref{lem:bound_alpha_mala_2}, we  provide a decomposition in $\gamma$ of $\alphamala_{\gamma}$ defined in \eqref{eq:def-alpha-MALA}. For any $x,z \in \rset^d$,  by \cite[Lemma 24]{durmus:moulines:saksman:2017}\footnote{Note that with the notation of \cite{durmus:moulines:saksman:2017}, MALA corresponds to HMC with only one leapfrog step and step size equals to $(2\gamma)^{1/2}$},  we have  that
\begin{equation}
\label{lem:durmus_moulines_saksman}
\alphamala_{\gamma}(x,z) = \sum_{k=2}^6 \gamma^{k/2} A_{k,\gamma}(x,z)
\end{equation}
where, setting $x_t = x+t\{-\gamma \nabla U(x) + \sqrt{2\gamma} z \}$,
\begin{align}
& A_{2,\gamma}(x,z)= 2 \int_0^1 \DD^2 U(x_t) [z^{\otimes 2}] (1/2-t) \rmd t \\
& A_{3,\gamma}(x,z)= 2^{3/2} \int_{0}^1 \DD^2 U(x_t) [z \otimes \nabla U(x)](t-1/4) \rmd t \,, \\
& A_{4,\gamma}(x,z)= -  \int_{0}^1 \DD^2 U(x_t)[ \nabla U(x)^{\otimes 2}] t \rmd t + (1/2) \normEq[2]{ \int_{0}^1 \DD^2 U(x_t) [z] \rmd t }  \\
& A_{5,\gamma}(x,z)=   -(1/2)^{1/2}\psEq{\int_{0}^1  \DD^2 U(x_t) [\nabla U(x)] \rmd t }{ \int_{0}^1 \DD^2 U(x_t) [z] \rmd t} \\
& A_{6,\gamma}(x,z)= (1/4) \normEq[2]{\int_{0}^1 \DD^2 U(x_t) [\nabla U(x)] \rmd t } \eqsp.
\end{align}
\begin{proof}[Proof of \Cref{lem:bound_alpha_mala_1}]
  Let $\bgamma >0$, $\gamma \in \ocint{0,\bgamma}$.
  Since $\int_{0}^1 \DD^2 U(x) [z^{\otimes 2}](1/2-t) \rmd t = 0$, we get setting  $x_t = x + t \{-\gamma \nabla U(x) + \sqrt{2\gamma} z \}$,  with $x,z \in \rset^d$,
\begin{multline}
\label{eq:decomposition-A-2}
A_{2,\gamma}(x,z) \\= \sqrt{\gamma} \int_0^1 \int_0^1 \DD^3 U (s x_t + (1-s) x) \parentheseDeux{z^{\otimes 2} \otimes \{ -\gamma^{1/2} \nabla U(x) + \sqrt{2} z \}} (1/2-t) t \rmd s \rmd t \eqsp.
\end{multline}
Therefore, we get using \eqref{lem:durmus_moulines_saksman}, \Cref{ass:regularity-U}, \Cref{ass:regularity-U_C3} $\abs{ab} \leq (a^2+b^2)/2$ and $\norm{\nabla U(x)} \leq \Ltt \norm{x}$,
\begin{equation}
\frac{  \abs{\alphamala_{\gamma}(x,z)}}{\gamma^{3/2}} \leq 2^{1/2}\Mtt \norm[3]{z} + \bgamma^{1/2} \Mtt \Ltt \norm{z}^2 \norm{x}  + 2 \Ltt^2(\norm[2]{z} + \norm[2]{x})[1\vee\bgamma^{1/2}\vee \bgamma \Ltt \vee (\bgamma \Ltt^{4/3})^{3/2}] \eqsp.
\end{equation}
The proof of \eqref{eq:lem:bound_alpha_mala_1} then easily follows using $\norm[3]{z} \leq \norm[2]{z} + \norm[4]{z}$, $\norm[2]{z} \norm{x} \leq \norm[4]{z} + \norm[2]{x}$ and setting
\begin{equation}
  \label{eq:def_C_1_bgamma}
  C_{1,\bgamma}= 2 (2^{1/2}\Mtt \vee \bgamma^{1/2} \Mtt \Ltt \vee 2 \Ltt^2 [1\vee\bgamma^{1/2}\vee \bgamma \Ltt \vee (\bgamma \Ltt^{4/3})^{3/2}]) \eqsp.
\end{equation}
\end{proof}

\begin{proof}[Proof of \Cref{lem:bound_alpha_mala_2}]
We now show \eqref{eq:lem:bound_alpha_mala_2}. Let  $x,z \in \rset^d$ satisfying $\norm{x} \geq \max(2\Ktt,\tKtt)$ and $\norm{z} \leq \norm{x}/(4\sqrt{2\gamma})$.
  Using \eqref{lem:durmus_moulines_saksman}, \Cref{ass:regularity-U}, $\abs{ab} \leq (a^2+b^2)/2$ and $\norm{\nabla U(x)} \leq \Ltt \norm{x}$, we get setting \[ A_{4,0,\gamma}(x,z)=  \int_{0}^1 \DD^2 U(x_t)  [\nabla U(x)^{\otimes 2}] t \rmd t \eqsp, \]
  \begin{multline}
    \label{eq:2}
        \alphamala_{\gamma}(x,z) \leq 2\gamma \Ltt \norm[2]{z} -\gamma^2 A_{4,0,\gamma}(x,z) \\+(2\gamma)^{3/2}\Ltt^2 \norm{z}\norm{x} + (\gamma^2/2) \Ltt^2 \norm[2]{z} + (\gamma^5/2)^{1/2} \Ltt^3 \norm{z}\norm{x} + (\gamma^3/4) \Ltt^4 \norm[2]{x} \eqsp.
  \end{multline}
By \Cref{lem:bounde_pertubhessian}, $\norm{x_t} \geq \norm{x}/2$ since $\norm{z} \leq \norm{x}/(4\sqrt{2\gamma})$ and $\gamma \leq \bgamma \leq \mtt^3/(4\Ltt^4) \leq 1/(4\Ltt)$. Therefore, \Cref{ass:curvature_U} and  \Cref{lem:quadratic_behaviour}  imply since   $\norm{x} \geq \max(2\Ktt,\tKtt)$ that
  \begin{equation}
    \label{eq:3}
    A_{4,0,\gamma}(x,z) \geq (\mtt/2)^3\norm[2]{x}  \eqsp.
  \end{equation}
  Combining this result with \eqref{eq:2}, we obtain using $\gamma \leq \bgamma \leq \mtt^3/(4\Ltt^4)$
  \begin{align}
    \alphamala_{\gamma}(x,z)
    & \leq 2\gamma \Ltt \norm[2]{z} -\gamma^2(\mtt^3/2^4) \norm[2]{x} \\
                   & \qquad +(2\gamma)^{3/2}\Ltt^2 \norm{z}\norm{x} + (\gamma^2/2) \Ltt^2 \norm[2]{z} + (\gamma^5/2)^{1/2} \Ltt^3 \norm{z}\norm{x}\eqsp,
  \end{align}
Since for any $a, b \geq 0$ and $\epsilon > 0$, $ab \leq  (\epsilon/2) a^2 + 1/(2\epsilon) b^2$,  we obtain
  \begin{align}
    &    \alphamala_{\gamma}(x,z) \leq \gamma \norm[2]{z} \Big\{2 \Ltt+ 2^{1/2} \Ltt^2 \epsilon^{-1} + (\gamma/2)\Ltt^2 + 2^{-3/2}  \gamma^{3/2} \Ltt^3 \epsilon^{-1}  \Big\}\\
&\qquad \qquad \qquad     + \norm[2]{x} \gamma^2 \parentheseDeux{\epsilon\defEns{2^{1/2} \Ltt^2 + 2^{-3/2} \bgamma^{1/2} \Ltt^3 } -\mtt^3/2^4} \eqsp.
  \end{align}
  Choosing $\epsilon = (\mtt^3/2^4) \defEnsLigne{2^{1/2} \Ltt^2 + 2^{-3/2} \bgamma^{1/2} \Ltt^3}^{-1}$ concludes the proof with
  \begin{equation}
\label{eq:def_C_2_bgamma}
C_{2,\bgamma} =     2 \Ltt+ 2^{1/2} \Ltt^2 \epsilon^{-1} + (\bgamma/2)\Ltt^2 + 2^{-3/2}  \bgamma^{3/2} \Ltt^3 \epsilon^{-1} \eqsp.
  \end{equation}
\end{proof}
Properties \ref{item:condition_ergo_I} and \ref{item:condition_ergo_II} follow from \Cref{lem:bound_alpha_mala_1} and \Cref{lem:bound_alpha_mala_2}

\subsection{Foster-Lyapunov drift condition}
\label{subsec:geom-ergodicity-mala}

To show that MALA satisfies a Lyapunov condition of the form \eqref{eq:def-discrete-drift}, we first show that
it holds for ULA from the following result.
\begin{proposition}
  \label{propo:super_lyap_ula}
  Assume \Cref{ass:regularity-U} and \Cref{ass:curvature_U} and let $\bgamma \in \ocint{0,\mtt/(4\Ltt^2)}$. Then, for any $\gamma \in \ocint{0,\bgamma}$, $x \in\rset^d$,
  \begin{equation}
    Q_{\gamma} V_{\bareta}(x) \leq \exp\parenthese{-\bareta \mtt \gamma \norm[2]{x}/4} V_{\bareta}(x) + \bdriftula \gamma \1_{\ball{0}{\rayula}}(x) \eqsp,
  \end{equation}
  where $V_{\bareta}$ is defined by \eqref{eq:def_V_eta}, $\bareta  = \mtt/16$, $\rayula = \max(\tKtt,4\sqrt{d/\mtt})$, $\tKtt$ is defined in \Cref{lem:quadratic_behaviour} and
  \begin{equation}
  \label{eq:coeffs_super_lyap_mala}
  \begin{aligned}
    \bdriftula &= \parentheseDeux{\bareta \defEns{ \mtt/4+   (1+16\bareta\bgamma)(4\bareta + 2 \Ltt + \bgamma \Ltt^2)} (\rayula)^2 +4 \bareta d  } \\
    &  \qquad \times \exp\parentheseDeux{\bgamma\bareta\defEns{\mtt/4+   (1+16\bareta\bgamma)(4\bareta + 2 \Ltt + \bgamma \Ltt^2)} (\rayula)^2 + 4\bareta\bgamma d}\eqsp.
  \end{aligned}
\end{equation}
\end{proposition}

\begin{proof}
  Let $\gamma \in \ocint{0,\bgamma}$. First note that $\bareta \gamma \leq \mtt \bgamma /16 \leq \mtt^2/(2^6\Ltt^2) \leq 1/8$, since $\Ltt \geq \mtt$, and therefore $1-4\bareta\gamma \geq 1/2$. In addition for any $x \in  \rset^d$, we have
  \begin{multline}
    \bareta \norm[2]{x-\gamma \nabla U(x) + \sqrt{2\gamma} z} -\norm[2]{z} /2 \\
    = -\frac{1-4\bareta\gamma}{2} \norm[2]{z-\frac{ 2(2\gamma)^{1/2}\bareta}{1-4\bareta\gamma}\{x-\gamma\nabla U(x)\}} + \frac{\bareta}{1-4\bareta\gamma} \norm[2]{x- \gamma \nabla U(x)}  \eqsp,
  \end{multline}
which implies since $1-4 \bareta\gamma > 0$ that
  \begin{align}
    \nonumber
    \Qgam V_{\bareta}(x) & = (2\uppi)^{-d/2}\int_{\rset^d} \exp\parenthese{    \bareta \norm[2]{x-\gamma \nabla U(x) + \sqrt{2\gamma} z} -\norm[2]{z} /2} \rmd z \\
    \label{eq:1:propo:super_lyap_mala}
    & =  (1-4\bareta\gamma)^{-d/2} \exp\parenthese{ \bareta(1-4\bareta\gamma)^{-1}\norm[2]{x- \gamma \nabla U(x)}} \eqsp.
  \end{align}
  We now distinguish the case when $\norm{x} \geq \rayula$ and $\norm{x} < \rayula$.

  By  \Cref{lem:quadratic_behaviour}, for any $x \in \rset^d$, $\norm{x} \geq \rayula \geq \tKtt$, using that $\bareta = \mtt/16$ and $\gamma \leq \bgamma \leq \mtt/(4\Ltt^2)$, we have
  \begin{multline}
    (1-4\bareta\gamma)^{-1}  \norm[2]{x- \gamma \nabla U(x)} -\norm[2]{x}\\
    \leq \gamma \norm[2]{x}(1-4\bareta \gamma)^{-1} \parenthese{4\bareta - \mtt + \gamma \Ltt^2} \leq -\gamma (\mtt/2) \norm[2]{x} (1-4\bareta\gamma)^{-1}\eqsp.
\end{multline}
Therefore, \eqref{eq:1:propo:super_lyap_mala} becomes
  \begin{align}
    \Qgam V_{\bareta}(x)
   & \leq   \exp\parenthese{ -\gamma \bareta (\mtt/2) (1-4\bareta\gamma)^{-1}\norm[2]{x} - (d/2)\log(1-4\bareta\gamma)} V_{\bareta}(x) \\
    & \leq \exp\parenthese{ \gamma \bareta\{- (\mtt/2) \norm[2]{x} + 4 d\}} V_{\bareta}(x)  \eqsp,
  \end{align}
  where we have used for the last inequality that $-\log(1-t) \leq 2t$ for $t \in \ccint{0,1/2}$ and $4 \bareta \gamma \leq 1/2$. The proof of the statement then follows since $\norm{x} \geq \rayula \geq 4 \sqrt{d/\mtt}$.

  In the case $\norm{x }< \rayula$, by \eqref{eq:1:propo:super_lyap_mala}, \Cref{ass:regularity-U} and since $(1-t)^{-1} \leq 1+4t$ for $t\in\ccint{0,1/2}$, we obtain
\begin{align}
    (1-4\bareta\gamma)^{-1}\norm[2]{x- \gamma \nabla U(x)} - \norm[2]{x}
    &\leq  \gamma    (1-4\bareta\gamma)^{-1}\{4\bareta + 2 \Ltt + \gamma \Ltt^2\}\norm[2]{x} \\
    &\leq  \gamma    (1+16\bareta\gamma)\{4\bareta + 2 \Ltt + \gamma \Ltt^2\}\norm[2]{x} \eqsp,
\end{align}
we have
  \begin{multline}
    \Qgam V_{\bareta}(x)/V_{\bareta}(x) \leq \rme^{-\bareta  \mtt \gamma \norm[2]{x}/4}  \\
    +  \exp\parentheseDeux{ \gamma \bareta \defEns{ \mtt/4+   (1+16\bareta\gamma)(4\bareta + 2 \Ltt + \gamma \Ltt^2)}\norm[2]{x} -(d/2)\log(1-4\bareta\gamma)} -1 \eqsp.
  \end{multline}
  The proof is then completed using that for any $t \geq 0$, $\rme^{t} -1 \leq t \rme^{t}$, for any $s \in \ccint{0,1/2}$, $-\log(1-s) \leq 2s$ and $4\bareta\gamma \leq 1/2$.
\end{proof}

Combining the previous result with \Cref{lem:bound_alpha_mala_2}, we can show that MALA satisfies \eqref{eq:def-discrete-drift}.

\begin{proposition}
  \label{propo:lyap_mala_total}
  Assume \Cref{ass:regularity-U} and \Cref{ass:curvature_U}. Then, there exist $\Gamma_{1/2} \geq  \Gamma > 0$ (given in \eqref{eq:def_Gamma_1_2}-\eqref{eq:def_Gamma} in the proof) such that for any $\bgamma \in \ocint{0,\Gamma}$, $\gamma \in\ocint{0,\bgamma}$ and $x \in\rset^d$,
  \begin{equation}
     R_{\gamma}V_{\bareta}(x) \leq (1-\tildem \gamma)V_{\bareta}(x)+ \bdriftmala \gamma \1_{\ball{0}{\raymaladrift}}(x) \eqsp,
  \end{equation}
  where $V_{\bareta}$ is defined by \eqref{eq:def_V_eta}, $R_{\gamma}$ is the Markov kernel of MALA defined by \eqref{eq:def-kernel-MALA},  $\bareta=\mtt/16$, $    \tildem  = \bareta \mtt (\raymaladrift)^2/16 $,
  \begin{align}
    \label{eq:def_const_drift_1}
    \raymaladrift  &= \max(2^4,2\Ktt,\rayula,\tKtt,4b_{1/2}^{1/2}/(\mtt \bareta)^{1/2})  \eqsp, \quad     b_{1/2} = C_{2,\Gamma_{1/2}} d + \sup_{u \geq 1} \{u \rme^{-u/2^7}\} \eqsp, \\
    \label{eq:def_b_drift}
    \bdriftmala & = \bdriftula +  \bareta \mtt (\raymaladrift)^2 \rme^{\bareta (\raymaladrift)^2} / 16  + C_{1,\bgamma}  \bgamma^{1/2} \defEns{d + \sqrt{3}d^2 + (\raymaladrift)^2}\eqsp,
  \end{align}
$\rayula,\bdriftula$ are defined in \Cref{propo:super_lyap_ula}  and $C_{1,\bgamma}$ and $C_{2,\Gamma_{1/2}}$ in \eqref{eq:def_C_1_bgamma} and \eqref{eq:def_C_2_bgamma} respectively.
\end{proposition}

We preface the proof by a technical result.

\begin{lemma}
  \label{lem_tail_chi2}
Let $\bgamma >0$ and $\gamma \in \ocint{0,\bgamma}$. Then, for any $c >0$ and $x \in \rset^d$, $\norm{x} \geq \sqrt{8\bgamma d/c}$,
\begin{equation}
  \int_{\rset^d \setminus \ball{0}{(c/\gamma)^{1/2}\norm{x}}} \varphibf(z) \rmd z \leq \exp(-c\norm{x}^2/(4\gamma)) \eqsp,
\end{equation}
where $\varphibf$ is the density of the $d$-dimensional standard Gaussian distribution with respect to the Lebesgue measure.
\end{lemma}

\begin{proof}[Proof of \Cref{lem_tail_chi2}]
Let $u>0$.
By \cite[Lemma 1]{laurent:massart:2000},
\begin{equation}
  \PP(\norm[2]{Z} \geq d + 2 \sqrt{du} + 2u) \leq
  \rme^{-u} \eqsp,
\end{equation}
where $Z$ is a $d$-dimensional standard Gaussian vector.
Note that for $t \geq d$, the equation $d + 2 \sqrt{du} + 2u =t$ has a unique non-negative  solution $u_t = 4^{-1}(\sqrt{2t-d}-\sqrt{d})^2$. In addition for $t \geq 8d$, we have that $u_t \geq t/2^2$ using $u_t = 2^{-1}(t-\sqrt{d(2t-d)})$ and $ s/2-\sqrt{2s-1} \geq 0$ for $s \geq 8$. Therefore, we get setting  $t=d + 2 \sqrt{du} + 2u$,
\begin{equation}
  \PP(\norm[2]{Z} \geq t) \leq \exp\parenthese{-u_t} \leq \exp(t/4) \eqsp.
\end{equation}
Choosing now $t = c \norm[2]{x}/\gamma \geq 8d$ for $\norm{x} \geq \sqrt{8d \bgamma/c}$ concludes the proof.
\end{proof}

\begin{proof}[Proof of \Cref{propo:lyap_mala_total}]
By \eqref{eq:diff-rula-rmala} and \Cref{propo:super_lyap_ula}, for any
  $\bgamma \leq \mtt/(4\Ltt^2)$, $\gamma \in \ocint{0,\bgamma}$ and $x  \in \rset^d$,
  \begin{align}
  &  \Rkerg V_{\bareta}(x)
\leq \Qgam V_{\bareta}(x) + V_{\bareta}(x)\int_{\rset^d} \{1-\min(1,\rme^{-\alphamala_{\gamma}(x,z)}\} \varphibf(z) \rmd z\\
 \label{eq:propo:lyap_mala_total_1}
  & \leq \rme^{-\bareta m \gamma \norm[2]{x}/4}  V_{\bareta}(x) + \bdriftula \gamma \1_{\ball{0}{\rayula}}(x) + V_{\bareta}(x)\int_{\rset^d} \{1-\min(1,\rme^{-\alphamala_{\gamma}(x,z)}\} \varphibf(z) \rmd z \eqsp,
  \end{align}
  where $\rayula$ and $\bdriftula$ are given in \Cref{propo:super_lyap_ula}. Let $\Upsilon \leq 1$ and
  \begin{equation}
      \label{eq:def_Gamma_1_2}
\Gamma_{1/2} = \min\parenthese{\Upsilon, \mtt^3/(4\Ltt^4), d^{-1}} \eqsp, \quad
  \ray_{1/2} = \max\parenthese{2^4, 2 \Ktt ,  \rayula,\tKtt} \eqsp,
\end{equation}
where  $\tKtt$ is given in \Cref{lem:quadratic_behaviour}.
Note that $\Gamma_{1/2} \leq \mtt/(4\Ltt^2)$ since $\mtt \leq \Ltt$ and $\ray_{1/2} \geq \sqrt{2^8\gamma d}$ for $\gamma \in \ocint{0,\Gamma_{1/2}}$.
Then,  by \Cref{lem:bound_alpha_mala_2} and \Cref{lem_tail_chi2} with $c = 1/2^{5}$,  for any $x \in \rset^d$, $\norm{x} \geq \ray_{1/2}$, $\bgamma \in \ocint{0,\Gamma_{1/2}}$ and
 $\gamma \in \ocint{0,\bgamma}$,
  \begin{align}
    \Rkerg V_{\bareta}(x) & \leq  \rme^{-\bareta m \gamma \norm[2]{x}/4}  V_{\bareta}(x) + V_{\bareta}(x) \defEns{C_{2,\Gamma_{1/2}} \, d \gamma  + \exp(-\norm[2]{x}/(2^7\gamma))} \\
    & \leq  \rme^{-\bareta m \gamma \norm[2]{x}/4}  V_{\bareta}(x) + V_{\bareta}(x) \gamma b_{1/2}  \eqsp,
  \end{align}
  where $ b_{1/2}$ is defined in \eqref{eq:def_const_drift_1}.
  For ease of notation, we simply denote $\raymaladrift$ by $\raymaladriftD$. Note that $\raymaladriftD = \max(\ray_{1/2},4 b_{1/2}^{1/2}/\sqrt{\bareta\mtt})$ and let
  \begin{equation}
  \label{eq:def_Gamma}
\Gamma = \min\parenthese{\Gamma_{1/2}, \mtt^3/(4\Ltt^4), d^{-1}, 4/\defEns{\mtt \bareta \raymaladriftD^2}} \eqsp.
  \end{equation}
  Then, since for any $t \in \ccint{0,1}$, $\rme^{-t} \leq 1-t/2$, we get for any $x \in \rset^d$, $\norm{x} \geq \raymaladriftD$, $\bgamma \in \ocint{0,\Gamma}$ and $\gamma \in \ocint{0,\bgamma}$,
  \begin{align}
    \nonumber
    \Rkerg V_{\bareta}(x)& \leq  \rme^{-\bareta \mtt \gamma \raymaladriftD^2 /4}  V_{\bareta}(x) + V_{\bareta}(x) \gamma b_{1/2} \\
    \label{eq:drift_mala_totla_2}
    & \leq \parentheseDeux{1-\gamma\defEns{\bareta \mtt \raymaladriftD^2 /8 -b_{1/2}}} V_{\bareta}(x)
 \leq \defEns{1-\gamma \bareta \mtt \raymaladriftD^2 /16} V_{\bareta}(x) \eqsp.
  \end{align}
  In addition, by \eqref{eq:propo:lyap_mala_total_1} and  \Cref{lem:bound_alpha_mala_1}, using that for any $t \in \rset$, $1-\min(1,\rme^{-t}) \leq \abs{t}$, for any $x \in \rset^d$, $\norm{x} \leq \raymaladriftD$, $\bgamma \in \ocint{0,\Gamma}$ and $\gamma \in \ocint{0,\bgamma}$,
  \begin{align}
    \Rkerg V_{\bareta}(x) & \leq V_{\bareta}(x) +\bdriftula \gamma \1_{\ball{0}{\raymala_3}}(x) + C_{1,\bgamma} \gamma^{3/2} \int_{\rset^d} \{\norm[2]{z}+\norm[2]{x} + \norm[4]{z}\} \varphibf(z) \rmd z \\
    & \leq (1-\gamma \bareta \mtt \raymaladriftD^2 /16) V_{\bareta}(x) +  \gamma \bareta \mtt \raymaladriftD^2 \rme^{\bareta \raymaladriftD^2} / 16 + \gamma \bdriftula \\
    &\phantom{-------------}+ C_{1,\bgamma} \gamma \bgamma^{1/2} \defEns{d + \sqrt{3}d^2 + \raymaladriftD^2} \eqsp,
  \end{align}
  Combining this result and \eqref{eq:drift_mala_totla_2} completes the proof.

\end{proof}

\subsection{Minorization condition}

We follow the same strategy as the proof of the Lyapunov drift condition for MALA
regarding the minorization condition \eqref{eq:def-minorization}. We first show it holds for ULA in the following
result.
\begin{proposition}
  \label{propo:small_set_ula}
  Assume \Cref{ass:regularity-U}. Then for any $\Rrm \geq 0$,  $x,y \in \rset^d$, $\norm{x}\vee \norm{y} \leq \Rrm$, and $\gamma \in \ocint{0,1/\Ltt}$ we have
  \begin{equation}
    \tvnorm{\updelta_x \Qgam^{\ceil{1/\gamma}} - \updelta_y \Qgam^{\ceil{1/\gamma}}}   \leq 2 (1-\varepsilonula(\ray)) \eqsp.
  \end{equation}
  with
  \begin{equation}
    \label{eq:def_varepsi_ula}
    \varepsilonula(\ray) = 2\Phibf\parenthese{-(1+1/\Ltt)^{1/2}(3\Ltt)^{1/2}\ray} \eqsp.
  \end{equation}
\end{proposition}

\begin{proof}
By \Cref{ass:regularity-U} for any $x,y \in \rset^d$,
 \[ \norm[2]{x-y-\gamma\{\nabla U(x) - \nabla U(y)\}} \leq (1+ \gamma \upkappa(\gamma)) \norm[2]{x-y} \]
 where $\upkappa(\gamma) =  (2 \Ltt+\Ltt^2 \gamma)$. The proof follows from  \cite[Corollary 5]{debortoli2018back}.
\end{proof}

We then use \Cref{lem:bound_alpha_mala_1} to obtain the following bounds on the total variation distance between the iterates of MALA and ULA starting from the same initial point. Combined with the previous result, this will allow us to use a perturbation argument to show \eqref{eq:def-minorization}.

\begin{lemma}
  \label{lem:diff_tv_MALA_ULA}
  Assume \Cref{ass:regularity-U} and let $\bgamma>0$. Then,  for any $x \in \rset^d$ and $\gamma \in \ocint{0,\bgamma}$, we have
  \begin{equation}
    \label{eq:1:lem:diff_tv_MALA_ULA}
    \tvnorm{\updelta_x \Qgam - \updelta_x \Rkerg}  \leq C_{1,\bgamma} \gamma^{3/2} (d+\sqrt{3} d^2+\norm[2]{x}) \eqsp.
  \end{equation}
  If in addition and \Cref{ass:curvature_U} holds and let $\bgamma\in\ocint{0,\mtt/(4\Ltt^2)}$. Then,  for any $x \in \rset^d$ and $\gamma \in \ocint{0,\bgamma}$,
\begin{align}
    \label{eq:2:lem:diff_tv_MALA_ULA}
    \tvnorm{\updelta_x \Qgam^{\ceil{1/\gamma}} - \updelta_x \Rkerg^{\ceil{1/\gamma}}} &\leq C_{1,\bgamma} \gamma^{1/2} (d+\sqrt{3}d^2+\norm[2]{x}+2\tbdriftula/\mtt) \eqsp,
  \end{align}
  where $C_{1,\bgamma}$ is defined in \eqref{eq:def_C_1_bgamma} and
    \begin{equation}
    \label{eq:def_tbdrift_ula}
\tbdriftula =  2d + [\max\parentheseLigne{\tKtt, 2\sqrt{(2d)/\mtt}}]^2 \parenthese{\bgamma \Ltt^2 + 2\Ltt + \mtt/2} \eqsp.
  \end{equation}
\end{lemma}

We preface the proof by a technical lemma.

\begin{lemma}
  \label{propo:lyap_mala}
  Assume \Cref{ass:regularity-U}, \Cref{ass:curvature_U} and let $\bgamma\in\ocint{0,\mtt/(4\Ltt^2)}$. Then, for any $\gamma \in \ocint{0,\bgamma}$ and $x\in\rset^d$, $    \int_{\rset^d} \norm[2]{y} Q_{\gamma}(x,\rmd y) \leq \defEns{1-(\mtt\gamma)/2} \norm[2]{x}
    +  \gamma \tbdriftula$,
  where $Q_{\gamma}$ is the Markov kernel of ULA defined in \eqref{eq:def_kernel_ula} and  $\tKtt$ is defined in \Cref{lem:quadratic_behaviour}, and $\tbdriftula$ is defined in \eqref{eq:def_tbdrift_ula}.
\end{lemma}

\begin{proof}
  Let $\gamma\in\ocint{0,\bgamma}$ and $x\in\rset^d$. By \Cref{ass:regularity-U}, we have
  \begin{equation}
    \int_{\rset^d} \norm[2]{y} Q_{\gamma}(x,\rmd y) \leq
    2\gamma d + \norm[2]{x}(1+\gamma^2 \Ltt^2) - 2\gamma\ps{\nabla U(x)}{x} \eqsp.
  \end{equation}
  Set $\ray = \max\parentheseLigne{\tKtt, 2\sqrt{(2d)/\mtt}}$.
  We distinguish the case when $\norm{x} \geq \ray$ and $\norm{x} < \ray$.
  If $\norm{x} \geq \ray \geq \tKtt$, by \Cref{lem:quadratic_behaviour}, and since $\gamma \leq \bgamma \leq \mtt/(4\Ltt^2)$, $\norm{x} \geq \ray \geq 2\sqrt{(2d)/\mtt}$,
  \begin{align}
    \int_{\rset^d} \norm[2]{y} Q_{\gamma}(x,\rmd y) &\leq
    \norm[2]{x} \parentheseDeux{1-\gamma \defEns{\mtt - \gamma \Ltt^2 - (2d)/\norm[2]{x}}} \leq \norm[2]{x} \defEns{1-\gamma \mtt /2} \eqsp.
  \end{align}
  If $\norm{x} < \ray$, we obtain
  \begin{equation}
    \int_{\rset^d} \norm[2]{y} Q_{\gamma}(x,\rmd y) \leq
    \norm[2]{x} \defEns{1-\gamma \mtt /2} +
    \gamma \norm[2]{x} \parenthese{\gamma \Ltt^2 + 2 \Ltt + \mtt/2} + 2\gamma d \eqsp,
  \end{equation}
  which concludes the proof.
\end{proof}

\begin{proof}
  Let $x \in \rset^d$  and $\gamma \in \ocint{0,\bgamma}$.
We first show that \eqref{eq:1:lem:diff_tv_MALA_ULA} holds and then use this result to prove \eqref{eq:2:lem:diff_tv_MALA_ULA}.
Let $f : \rset^d \to \rset$ be a bounded and measurable function. Then, by \eqref{eq:diff-rula-rmala}, we have
\begin{align}
  &\abs{\Qgam f(x) - \Rkerg f(x)} \\
  & \qquad  = \Big| \int_{\rset^d}\{f(x-\gamma \nabla U(x) + \sqrt{2 \gamma} z) - f(x)\}  \{1 - \min(1,\rme^{-\alphamala_{\gamma}(x,z)}) \} \varphibf(z) \rmd z \Big| \\
  & \qquad  \leq 2 \norm{f}_{\infty} \int_{\rset^d} \abs{1 - \min(1,\rme^{-\alphamala_{\gamma}(x,z)}) } \varphibf(z) \rmd z   \leq  2 \norm{f}_{\infty} \int_{\rset^d} \abs{\alphamala_{\gamma}(x,z)}  \varphibf(z) \rmd z \eqsp.
\end{align}
The conclusion of \eqref{eq:1:lem:diff_tv_MALA_ULA} then follows from an application of \Cref{lem:bound_alpha_mala_1}.

We now turn to the proof of \eqref{eq:2:lem:diff_tv_MALA_ULA}. Consider the following decomposition
  \begin{equation}
    \updelta_x \Qgam^{\ceil{1/\gamma}} - \updelta_x \Rkerg^{\ceil{1/\gamma}} = \sum_{k=0}^{\ceil{1/\gamma}-1} \updelta_x \Qgam^k \{\Qgam - \Rkerg\} \Rkerg^{\ceil{1/\gamma}-k-1} \eqsp.
  \end{equation}
  Therefore using the triangle inequality, we obtain that
  \begin{equation}
            \label{eq:4}
        \tvnorm{\updelta_x \Qgam^{\ceil{1/\gamma}} - \updelta_x \Rkerg^{\ceil{1/\gamma}}} \leq \sum_{k=0}^{\ceil{1/\gamma}-1} \tvnorm{ \updelta_x \Qgam^k \{\Rkerg - \Qgam\} \Rkerg^{\ceil{1/\gamma}-k-1}} \eqsp.
      \end{equation}
      We now bound each term in the sum. Let $k \in \{0,\ldots,\ceil{1/\gamma}-1\}$ and $f : \rset^d \to \rset$ be a bounded and measurable function. By   \eqref{eq:1:lem:diff_tv_MALA_ULA}, we obtain that
      \[ \abs{ \updelta_x \{\Rkerg - \Qgam\} \Rkerg^{\ceil{1/\gamma}-k-1} f} \leq C_{1,\bgamma} \norm{f}_{\infty} \gamma^{3/2} \{d+\sqrt{3}d^2+\norm[2]{x}\} \]
      and therefore using \Cref{propo:lyap_mala}, we get
      \begin{equation}
\abs{       \updelta_x \Qgam^k   \{\Rkerg - \Qgam\} \Rkerg^{\ceil{1/\gamma}-k-1} f} \leq C_{1,\bgamma} \norm{f}_{\infty} \gamma^{3/2} \{d+\sqrt{3}d^2+(1-\mtt\gamma/2)^k \norm[2]{x} + 2\tbdriftula/\mtt \} \eqsp.
     \end{equation}
     Plugging this result in         \eqref{eq:4}, we obtain
     \begin{align}
       \tvnorm{\updelta_x \Qgam^{\ceil{1/\gamma}} - \updelta_x \Rkerg^{\ceil{1/\gamma}}} &\leq C_{1,\bgamma} \gamma^{3/2} \sum_{k=0}^{\ceil{1/\gamma}-1}   \{d+\sqrt{3}d^2+(1-\mtt\gamma/2)^k \norm[2]{x} + 2\tbdriftula/\mtt \} \\
       &\leq C_{1,\bgamma} \gamma^{1/2} \{d+\sqrt{3}d^2+\norm[2]{x}+2\tbdriftula/\mtt\} \eqsp,
     \end{align}
     which concludes the proof.
\end{proof}

\begin{proposition}
  \label{propo:small_set_mala}
  Assume \Cref{ass:regularity-U} and \Cref{ass:curvature_U}. Then for any $\ray \geq 0$ there exists $\tGamma_{\ray} > 0$ (given in \eqref{eq:def_bgamma_ray} in the proof), such that for any $x,y \in \rset^d$, $\norm{x}\vee \norm{y} \leq \ray$, and $\gamma \in \ocintLigne{0,\tGamma_{\ray}}$ we have
  \begin{equation}
\label{eq:small_set_mala_propo}
    \tvnorm{\updelta_x \Rkerg^{\ceil{1/\gamma}} - \updelta_y \Rkerg^{\ceil{1/\gamma}}}   \leq 2 (1-\varepsilonula(\ray)/2) \eqsp,
  \end{equation}
  where $\varepsilonula(\ray)$ is defined in \eqref{eq:def_varepsi_ula}.
\end{proposition}

\begin{proof}
First note that for any $x,y \in \rset^d$, $\gamma >0$, by the triangle inequality, we obtain
  \begin{multline}
    \label{eq:decomposiiton_small_set_mala}
       \tvnorm{\updelta_x \Rkerg^{\ceil{1/\gamma}} - \updelta_y \Rkerg^{\ceil{1/\gamma}}} \leq        \tvnorm{\updelta_x \Rkerg^{\ceil{1/\gamma}} - \updelta_x \Qgam^{\ceil{1/\gamma}}}\\ + \tvnorm{\updelta_x \Qgam^{\ceil{1/\gamma}} - \updelta_y \Qgam^{\ceil{1/\gamma}}} + \tvnorm{\updelta_y \Rkerg^{\ceil{1/\gamma}} - \updelta_y \Qgam^{\ceil{1/\gamma}}} \eqsp.
     \end{multline}
     We now give some bounds for each term on the right hand side for any $x,y \in \rset^d$,  $\norm{x}\vee \norm{y} \leq \ray$ for a fixed $\ray \geq 0$ and $\gamma \leq  1/\Ltt$.
     By \Cref{propo:small_set_ula},  for any $x,y \in \rset^d$,  $\norm{x}\vee \norm{y} \leq \ray$ and $\gamma \leq  1/\Ltt$,
     \begin{equation}\label{eq:bound_small_ULA_proof_small_MALA}
       \tvnorm{\updelta_x \Qgam^{\ceil{1/\gamma}} - \updelta_y \Qgam^{\ceil{1/\gamma}}} \leq
       2(1-\varepsilonula(\ray)) \eqsp.
     \end{equation}
     In addition,
consider for $\tUpsilon >0$, $\tGamma_{1/2} = \tUpsilon \wedge \mtt/(4\Ltt^2)$. By \Cref{lem:diff_tv_MALA_ULA},  for any $\gamma \in \ocintLigne{0,\tGamma_{1/2}}$, and $x \in \rset^d$, $\norm{x} \leq \ray$,
  \begin{equation}
    \tvnorm{\updelta_x \Qgam^{\ceil{1/\gamma}} - \updelta_x \Rkerg^{\ceil{1/\gamma}}} \leq  C_{1,\tGamma_{1/2}}\gamma^{1/2}(d+\sqrt{3} d^2+\ray^2+2\tbdriftulatGamma/\mtt) \eqsp.
  \end{equation}

  Combining this result with \eqref{eq:bound_small_ULA_proof_small_MALA} in \eqref{eq:decomposiiton_small_set_mala}, we obtain that for any $\ray \geq 0$, $\gamma \in \ocintLigne{0,\tGamma_{1/2}}$, and  $x,y\in \rset^d$, $\norm{x}\vee\norm{y} \leq \ray$,
  \begin{equation}
    \norm{\updelta_x \Rkerg^{\ceil{1/\gamma}} - \updelta_y \Rkerg^{\ceil{1/\gamma}}} \leq 2(1-\varepsilonula(\ray)) + 2 C_{1,\tGamma_{1/2}}\gamma^{1/2}(d+\sqrt{3} d^2+\ray^2+2\tbdriftulatGamma/\mtt) \eqsp.
  \end{equation}
  Therefore, we obtain that for any $x,y \in \rset^d$, $\norm{x} \vee \norm{y} \leq \ray$, $\gamma \in \ocintLigne{0,\tGamma_{\ray}}$, \eqref{eq:small_set_mala_propo} holds  taking
  \begin{equation}
    \label{eq:def_bgamma_ray}
\tGamma_{\ray} = \tGamma_{1/2} \wedge \parentheseDeux{\frac{\varepsilon(\ray)}{2 C_{1,\tGamma_{1/2}}(d+\sqrt{3} d^2+\ray^2+2\tbdriftulatGamma/\mtt)}}^2 \eqsp.
  \end{equation}
\end{proof}

\subsection{Proof of \Cref{theo:V-geom_ergo_MALA_just_C2}}
\label{sec:proof:theo:V-geom_ergo_MALA_just_C2}
It is easy to prove that any compact set of $\rset^d$ with positive Lebesgue measure is a small set for the Markov kernel associated with MALA \eqref{eq:def-kernel-MALA} and in particular that $\Qmala_{\gamma}$ is strongly aperiodic and irreducible; see \eg~\cite[Example 9.1.5]{douc:moulines:priouret:soulier:2018}. The proof follows from \cite[Theorem 15.2.4.]{douc:moulines:priouret:soulier:2018}

\subsection{Proof of \Cref{theo:V-geom_ergo_MALA}}
\label{sec:proof:theo:V-geom_ergo_MALA}

\Cref{propo:lyap_mala_total} shows that there exist $\Gamma_{1/2} \geq  \Gamma > 0$ (given in \eqref{eq:def_Gamma_1_2}-\eqref{eq:def_Gamma}) such that for any $\bgamma \in \ocint{0,\Gamma}$, $\gamma \in\ocint{0,\bgamma}$ and $x \in\rset^d$,
  \begin{equation}
     R_{\gamma}V_{\bareta}(x) \leq (1-\tildem \gamma)V_{\bareta}(x)+ \bdriftmala \gamma  \eqsp,
  \end{equation}
  where $V_{\bareta}$ is defined by \eqref{eq:def_V_eta}, $R_{\gamma}$ is the Markov kernel of MALA defined by \eqref{eq:def-kernel-MALA},  $\bareta=\mtt/16$, $    \tildem, \bdriftmala$ are specified in the statement of \Cref{propo:lyap_mala_total}.
  Using \cite[Lemma 14.1.10]{douc:moulines:priouret:soulier:2018}, we obtain $\pi(V_{\bareta}) \leq A_{\bgamma,\bareta} = \bdriftmala/\tildem$. Therefore, we get
  \begin{equation}
    \label{eq:def_bA_final}
    \pi(V_{\bareta}) \leq \barA_{\bareta} = \inf_{\bgamma \in \ocint{0,\Gamma}} A_{\bgamma,\bareta} \eqsp.
  \end{equation}
We now show \eqref{eq:V-geom_ergo_MALA}.
  Using $1-t \leq \rme^{-t}$ for $t \in \rset$ and setting $\lambdaFL = \rme^{-\tildem}<1$,  an easy induction implies that for any $\gamma \in \ocint{0,\bgamma}$,
\begin{equation}\label{eq:def-discrete-drift_2}
  \RKer_\step^{\ceil{1/\gamma}} \lV_{\bareta} \leq \lambdaFL \lV_{\bareta} +  \bdriftmala(1+\bgamma) \eqsp.
\end{equation}

Set now
\begin{equation}
  \label{eq:def_M_gamma}
M_{\bgamma}  = \parenthese{\frac{4\bdriftmala (1+\bgamma)}{1-\lambda}} \vee 1 \eqsp, \quad \ray_{\bgamma} = (\log(M_{\bgamma})/\bareta)^{1/2} \eqsp.
\end{equation}
Note that $\ball{0}{\ray_{\bgamma}}  = \{\lV_{\bareta} \leq M_{\bgamma}\}$, $\bgamma \mapsto \bdriftmala (1+\bgamma)$ and $\bgamma \mapsto \ray_{\bgamma}$ are increasing on $\rset_+$. Then, $\tGamma_{\ray_{\bgamma}} \geq \tGamma_{\ray_{\Gamma}}$ for $\bgamma \leq \Gamma$ where $\tGamma_{\ray}$ is defined in \eqref{eq:def_bgamma_ray}, and \Cref{propo:small_set_mala}  implies setting
\begin{equation}
  \label{eq:def_bGamma_final}
  \bGamma = \Gamma \wedge \tGamma_{\ray_{\Gamma}} \eqsp,
\end{equation}
that for any $\bgamma \in \ocint{0,\bGamma}$,  any $x,y\in \{\lV_{\bareta} \leq M_{\bgamma}\}$, and $\gamma \in \ocint{0,\bgamma}$,
$\tvnorm{\updelta_{x} \Qmala^{\ceil{1/\gamma}} - \updelta_{y} \Qmala^{\ceil{1/\gamma}}} \leq 2(1-\varepsilon(\ray_{\bgamma}))$. As a result,
\cite[Theorem~19.4.1]{douc:moulines:priouret:soulier:2018} applied to $\Qmala^{\ceil{1/\gamma}}$ shows   that for any $x \in\rset^d$, $n \in \nset$,
\begin{equation}
 \Vnorm[\lV_{\bareta}]{\updelta_x \Qmala^{n \ceil{1/\gamma}}_{\gamma} - \pi} \leq C_{\bgamma} \{ \lV_{\bareta}(x) + \pi(\lV_{\bareta}) \} \rho^n_{\bgamma} \eqsp,
  \end{equation}
  where
  \begin{equation}
    \label{eq:cst_bornes_mala}
      \begin{aligned}
        &\log \rho_{\bgamma}  = \frac{\log(1-2^{-1}\varepsilon(\ray_{\bgamma})) \log\bar\lambda} { \bigl(\log(1-2^{-1}\varepsilon(\ray_{\bgamma})) +
          \log\bar\lambda-\log {\barbdriftmalabeta}\bigr) } \eqsp ,\\
        &\bar\lambda = \lambda +(1-\lambda)/2 \eqsp, \qquad
 {\barbdriftmala} = \lambda \bdriftmala + M_{\bgamma}
        \eqsp,   \\
        &C_{\bgamma}  = \rho_{\bgamma}^{-1}\{\lambda+1\}\{1+{\barbdriftmala}/[(1-2^{-1}\varepsilon(\ray_{\bgamma})(1-\bar\lambda)]\} \eqsp.
      \end{aligned}
\end{equation}



\section{Extension to \Cref{ass:curvature_U_alpha}}
\label{sec:extens-cref}

Define for any $x \in \rset^d$ and $\etabeta >0$,
\begin{equation}
  \label{eq:lyap_sub_quad_beta}
\Wbeta_{\etabeta}(x) = \exp\{\etabeta (1+\norm[2]{x})^{1/2}\} \eqsp.  
\end{equation}

\begin{theorem}
  \label{theo:V-geom_ergo_MALA_just_C2_beta}
  Assume \Cref{ass:regularity-U} and \Cref{ass:curvature_U_alpha}$(\beta)$ for $\beta \in\coint{0,1}$.
Then, there exists $\Gamma_{\beta} >0$ (defined in \eqref{v_beta:eq:def_Gamma}) such for any $\gamma \in \ocint{0,\Gamma_{\beta}}$, there exist $C_{\gamma,\beta} \geq 0$ and $\rho_{\gamma,\beta}\in\coint{0,1}$ satisfying for any $x \in \rset^d$,
  \begin{equation}
    \label{eq:50}
    \Vnorm[\Wbeta_{\baretabeta}]{\updelta_x \Rmala_{\gamma} - \pi} \leq C_{\gamma,\beta} \rho_{\gamma,\beta}^{k}\Wbeta_{\baretabeta}(x) \eqsp,
  \end{equation}
  where $  \baretabeta=\mttbeta/2^5$.
\end{theorem}
\begin{proof}
  The proof is identical to the one of \Cref{theo:V-geom_ergo_MALA_just_C2} using \Cref{v_beta:propo:lyap_mala_total} below in place of \Cref{propo:lyap_mala_total}. Therefore it is omitted. 
\end{proof}

\begin{theorem}
  \label{theo:V-geom_ergo_MALA_beta}
  Assume \Cref{ass:regularity-U}, \Cref{ass:regularity-U_C3} and \Cref{ass:curvature_U_alpha}$(\beta)$ for $\beta \in\coint{0,1}$.
Then, there exist $\bGamma_{\beta},\barA_{\beta} > 0$ (defined in \eqref{eq:def_bGamma_final_v_beta} and \eqref{v_beta:eq:def_bA_final}), such that for any $\bgamma \in\ocint{0,\bGamma_{\beta}}$, there exist $C_{\bgamma,\beta} \geq 0$ and $\rho_{\bgamma,\beta}\in\coint{0,1}$ (given in \eqref{eq:cst_bornes_mala_v_beta}) satisfying for any $x \in \rset^d$ and $\gamma \in\ocint{0,\bgamma}$,
  \begin{equation}
    \label{v_beta:eq:5}
    \Vnorm[\Wbeta_{\baretabeta}]{\updelta_x \Rmala_{\gamma} - \pi} \leq C_{\bgamma,\beta} \rho_{\bgamma,\beta}^{\gamma k} \{ \Wbeta_{\baretabeta}(x) + \pi(\Wbeta_{\baretabeta})\} \eqsp,
  \end{equation}
  where $  \baretabeta=\mttbeta/2^5$ and $\pi(\Veta_{\bareta}) \leq \barA_{\beta}$.
\end{theorem}
\begin{proof}
  The proof is postponed to \Cref{sec:proof:theo:V-geom_ergo_MALA_v_beta}.
\end{proof}

\subsection{Bounds on the acceptance ratio}

\begin{lemma}
  \label{lem:bound_alpha_mala_2_alpha}
  Assume \Cref{ass:regularity-U} and \Cref{ass:curvature_U_alpha}$(\beta)$ for $\beta\in\coint{0,1}$. Then, for any $\bgamma \leq (4\Ltt)^{-1} \wedge (\mttbeta^3/(2^4 \Lttbeta^4))$, there exists an explicit constant (see \eqref{eq:def_C_2_bgamma_beta}) $C_{2,\bgamma,\beta} < \infty$ such that for any  $\gamma \in \ocint{0 ,\bgamma}$, $x,z \in\rset^d$, $\norm{x} \geq \max(2 \Kttbeta, \tKttbeta,\bKttbeta)$, $\tKttbeta$ given in \Cref{lem:quadratic_behaviour_alpha} and $\bKttbeta$ in \Cref{lem:upper_bound_behaviour_alpha}, and $\norm{z} \leq \norm{x}/(4\sqrt{2\gamma})$, it holds
  \begin{equation}
          \label{eq:lem:bound_alpha_mala_2_alpha}
    \alphamala_{\gamma}(x,z) \leq       C_{2,\bgamma}\gamma\norm[2]{z} \eqsp.
  \end{equation}
\end{lemma}

\begin{proof}[of \Cref{lem:bound_alpha_mala_2}]
  We show \eqref{eq:lem:bound_alpha_mala_2_alpha} using the decomposition \eqref{lem:durmus_moulines_saksman}. Let  $\bgamma \leq (4\Ltt)^{-1} \wedge (\mttbeta^3/(2^4 \Lttbeta^4))$, $\gamma \in \ocint{0,\bgamma}$, $x,z \in \rset^d$ satisfying $\norm{x} \geq \max(2 \Kttbeta, \tKttbeta,\bKttbeta)$ and $\norm{z} \leq \norm{x}/(4\sqrt{2\gamma})$.
  Note that by \Cref{lem:bounde_pertubhessian}, for any $t \in \ccint{0,1}$, $\norm{x_t} \geq \norm{x}/2 \geq \Kttbeta$. Therefore, using \eqref{lem:durmus_moulines_saksman}, \Cref{ass:regularity-U}, \Cref{ass:curvature_U_alpha}$(\beta)$, \Cref{lem:upper_bound_behaviour_alpha},   and $\abs{ab} \leq (a^2+b^2)/2$, we get setting \[ A_{4,0,\gamma}(x,z)=  \int_{0}^1 \DD^2 U(x_t)  [\nabla U(x)^{\otimes 2}] t \rmd t \eqsp, \]
  \begin{multline}
    \label{eq:2_beta}
        \alphamala_{\gamma}(x,z) \leq 2\gamma \Ltt \norm[2]{z} -\gamma^2 A_{4,0,\gamma}(x,z) +2(2\gamma)^{3/2}\Lttbeta^2 \norm{z}\norm{x}/(1+\norm[3\beta/4]{x})^2 + (\gamma^2/2) \Ltt^2 \norm[2]{z} \\+ (2\gamma^5)^{1/2} \Lttbeta^3 \norm{z}\norm{x}/(1+\norm[3\beta/4]{x})^3 + \gamma^3 \Lttbeta^4 \norm[2]{x} /(1+\norm[3\beta/4]{x})^4\eqsp.
  \end{multline}
  Using $\norm{x_t} \geq \norm{x}/2 \geq \Kttbeta$ again, \Cref{ass:curvature_U_alpha}$(\beta)$ and  \Cref{lem:quadratic_behaviour_alpha} imply since $\norm{x} \geq \max(2 \Kttbeta, \tKttbeta,\bKttbeta)$ that
  \begin{equation}
    \label{eq:3_beta}
    A_{4,0,\gamma}(x,z) \geq (\mttbeta/2)^3\norm[2]{x}/(1+\norm[\beta]{x})^3  \eqsp.
  \end{equation}
Note that using that $(1+a)^{3/4} \leq 1 + a^{3/4}$ for $a \geq 0$, therefore $(1+\norm[\beta]{x})^{-3} \geq (1+\norm[3\beta/4]{x})^{-4}$.   Combining this result with \eqref{eq:3_beta} in \eqref{eq:2_beta}, we obtain using $\gamma \leq \bgamma \leq \mttbeta^3/(2^4\Lttbeta^4)$
  \begin{align}
&    \alphamala_{\gamma}(x,z)
  \leq 2\gamma \Ltt \norm[2]{z} -\gamma^2(\mttbeta^3/2^4) \norm[2]{x}/(1+\norm[\beta]{x})^3 \\
     &  +2(2\gamma)^{3/2}\Lttbeta^2 \norm{z}\norm{x}/(1+\norm[3\beta/4]{x})^2 + (\gamma^2/2) \Ltt^2 \norm[2]{z} + (2\gamma^5)^{1/2} \Lttbeta^3 \norm{z}\norm{x}/(1+\norm[3\beta/4]{x})^{3}\eqsp,
  \end{align}
Since for any $a, b \geq 0$ and $\epsilon > 0$, $ab \leq  (\epsilon/2) a^2 + 1/(2\epsilon) b^2$, and $(1+\norm[\beta]{x})^{-3} \geq (1+\norm[3\beta/4]{x})^{-4}$,  we obtain
  \begin{align}
    &    \alphamala_{\gamma}(x,z) \leq \gamma \norm[2]{z} \defEns{2 \Ltt+ 2^{3/2} \Lttbeta^2 \epsilon^{-1} + (\gamma/2)\Ltt^2 + 2^{-1/2}  \gamma^{3/2} \Lttbeta^3 \epsilon^{-1}  }\\
&\qquad \qquad \qquad     + \norm[2]{x} \gamma^2 \parentheseDeux{\epsilon\defEns{2^{3/2} \Lttbeta^2 + 2^{-1/2} \bgamma^{1/2} \Lttbeta^3 } -\mttbeta^3/2^4}/(1+\norm[\beta]{x})^{3} \eqsp.
  \end{align}
  Choosing $\epsilon = (\mttbeta^3/2^4) \defEnsLigne{2^{3/2} \Lttbeta^2 + 2^{-1/2} \bgamma^{1/2} \Lttbeta^3}^{-1}$ concludes the proof with
  \begin{equation}
\label{eq:def_C_2_bgamma_beta}
C_{2,\bgamma,\beta} =     2 \Ltt+ 2^{3/2} \Lttbeta^2 \epsilon^{-1} + (\bgamma/2)\Ltt^2 + 2^{-1/2}  \bgamma^{3/2} \Lttbeta^3 \epsilon^{-1} \eqsp. 
  \end{equation}
\end{proof}

\subsection{Lyapunov drift condition}

\begin{proposition}
  \label{propo:super_lyap_beta_ula}
  Assume \Cref{ass:regularity-U}, \Cref{ass:curvature_U_alpha}$(\beta)$ for $\beta \in\coint{0,1}$ and let $\bgamma \in \ocint{ 0,\mttbeta/(2^5\Lttbeta^2)}$. Then for any $\gamma \in \ocintLigne{0,\bgamma}$ and $x \in\rset^d$,
  \begin{equation}
    \label{eq:propo:super_lyap_beta_ula}
    \Qula_{\gamma} \Wbeta_{\baretabeta}(x) \leq \exp\defEns{-\frac{\gamma \mttbeta \baretabeta \norm[]{x}}{16(1+\norm[\beta]{x})}} \Wbeta_{\baretabeta}(x) + \gamma \bdriftulabeta \1_{\ball{0}{\rayulabeta}}(x) \eqsp, 
  \end{equation}
  where
  \begin{align}
    \label{eq:def_rayulabeta}
    \baretabeta & =  \mttbeta/2^5\eqsp,
                  \qquad     \rayulabeta  = 1 \vee \bKttbeta \vee \tKttbeta \vee[2^5d/\mttbeta]^{1/(2-\beta)}\eqsp, \\
    \bdriftulabeta &= \defEns{ \baretabeta( \Ltt(1+ \Ltt/2)\rayulabeta^2 +  d+\baretabeta)+\frac{ \mttbeta \baretabeta (1+\rayulabeta^2)^{1/2}}{16(1+(\rayulabeta)^{\beta})}} \rme^{\gamma \baretabeta( \Ltt(1+ \Ltt/2)\rayulabeta^2 +  d+\baretabeta)}   \eqsp, 
  \end{align}
$  \bKttbeta$ and $ \tKttbeta$ are defined in \Cref{lem:upper_bound_behaviour_alpha} and \Cref{lem:quadratic_behaviour_alpha} respectively.
\end{proposition}

\begin{proof}
  Let $x \in\rset^d$ and $\gamma \in \ocint{0,\bgamma}$.
  By definition \eqref{eq:def_kernel_ula} and  \cite[Proposition 5.5.1, (5.4.1)]{bakry:gentil:ledoux:2014},  $\updelta_x\Rula_{\gamma}$ satisfies a log-Sobolev inequality and since $ y \mapsto (1+\norm[2]{y})^{1/2}$ is $1$-Lipschitz we have
  \begin{align}
    \label{eq:9_propo:super_lyap_beta_ula_0} 
    \txts    \Qula_{\gamma}\Wbeta_{\baretabeta}(x) & \leq \exp\parenthese{\gamma (\baretabeta)^2 + \baretabeta \int_{\rset^d}(1+\norm[2]{y})^{1/2} \Qula_{\gamma}(x,\rmd y)} \\
    \label{eq:9_propo:super_lyap_beta_ula} 
 & \leq \exp\defEns{\gamma (\baretabeta)^2 + \baretabeta \parenthese{\int_{\rset^d}(1+\norm[2]{y}) \Qula_{\gamma}(x,\rmd y)}^{1/2}}
                                                 \eqsp,
  \end{align}
  where we have used the Cauchy-Schwarz inequality in the last line. We then bound the second term.
  We have by \Cref{lem:bound_alpha_mala_2_alpha} and \Cref{lem:upper_bound_behaviour_alpha} for $x \not \in \ball{0}{1\vee\tKttbeta\vee\bKttbeta}$,
  \begin{align}
    \label{eq:10}
&    \txts (   \int_{\rset^d}(1+\norm[2]{y})^{1/2} \Qula_{\gamma}(x,\rmd y) )^{1/2}  = (1+\norm[2]{x- \gamma \nabla U(x)} + 2 \gamma d) ^{1/2} \\
        \label{eq:10_propo:super_lyap_beta_ula}
    & \leq (1+\norm[2]{x})^{1/2}(1-2^{-1}\gamma \mttbeta /(1+\norm[\beta]{x})+ 4 \Lttbeta^2\gamma^2/(1+\norm[\beta]{x})^2 + 2\gamma d/\norm[2]{x})^{1/2}\eqsp,
  \end{align}
  where we have used that $1 \geq t^2/(1+t^2) \geq 1/2$ for $t \geq 1$.
  Since $1/(1+\norm[\beta]{x}) \geq 1/(1+\norm[\beta]{x})^2$ and $\gamma \leq \bgamma \leq \mttbeta/(2^{5}\Lttbeta^2)$, we get  $-\mttbeta /(1+\norm[\beta]{x})+ 2^{5} \Lttbeta^2\gamma/(1+\norm[\beta]{x})^2 \leq 0$ for any $x \in\rset^d$. In addition, for $x \in\rset^d$, $\norm{x} \geq 1 \vee [2^{5} d/\mttbeta]^{1/(2-\beta)}$, $\mttbeta \geq 2^{5} d \norm[\beta-2]{x} \geq 2^{4} d (1+\norm[\beta]{x})/\norm[2]{x}$. Combining these two results in \eqref{eq:10_propo:super_lyap_beta_ula}, we obtain that for any $x \not \in \ball{0}{\rayulabeta}$,
  with $\rayulabeta$ given in \eqref{eq:def_rayulabeta},
  \begin{multline}
    \label{eq:6_propo:super_lyap_beta_ula}
    \txts (   \int_{\rset^d}(1+\norm[2]{y})^{1/2} \Qula_{\gamma}(x,\rmd y) )^{1/2} 
    \leq (1+\norm[2]{x})^{1/2}(1 - 4^{-1}\mttbeta\gamma/(1+\norm[\beta]{x}))^{1/2}\\ \leq (1+\norm[2]{x})^{1/2}(1-8^{-1}\gamma \mttbeta /(1+\norm[\beta]{x})) \eqsp,
   \end{multline}
   where we have used that for any $t \coint{0,1}$, $(1-t)^{1/2} \leq 1 - t/2$. Since $t \mapsto t/(1+t^{\beta})$ is non-decreasing on $\rset_+$, we get for any $x \not \in \ball{0}{1}$, $16^{-1} \mttbeta \norm{x}/(1+\norm[\beta]{x}) \geq 2^{-5}\mttbeta \geq \baretabeta$. This result with \eqref{eq:9_propo:super_lyap_beta_ula} and \eqref{eq:6_propo:super_lyap_beta_ula} imply that for any $x \not \in \ball{0}{\rayulabeta}$, 
   \begin{equation}
     \label{eq:11_propo:super_lyap_beta_ula}
      \Qula_{\gamma}\Wbeta_{\baretabeta}(x) \leq \exp\parenthese{-16^{-1}\gamma \mttbeta \baretabeta (1+\norm[2]{x})^{1/2}/(1+\norm[\beta]{x})} \Wbeta_{\baretabeta}(x) \eqsp. 
    \end{equation}
    We now consider the case $x \in \ball{0}{\rayulabeta}$. First note by the Cauchy-Schwarz inequality, $\norm{\nabla U(x)} \leq \Ltt \norm{x}$ under \Cref{ass:regularity-U}, we have
      \begin{equation}
    \label{eq:15}
        \int_{\rset^d}(1+\norm[2]{y})^{1/2} \Qula_{\gamma}(x,\rmd y)  \leq (1+(1+\gamma \Ltt(2+\gamma \Ltt))\norm[2]{x} + 2 \gamma d)^{1/2}  \eqsp.
      \end{equation}
      Therefore using that $(1+a_1)^{1/2} - (1+a_2)^{1/2} \leq \abs{a_1-a_2}/2$ for $a_1,a_2 \geq 0$, we get
      \begin{equation}
        \label{eq:15_propo:super_lyap_beta_ula}
        \int_{\rset^d}(1+\norm[2]{y})^{1/2} \Qula_{\gamma}(x,\rmd y) - (1+\norm[2]{x})^{1/2} \leq 2^{-1}\gamma \Ltt(2+\gamma \Ltt)\norm[2]{x} + \gamma d \eqsp.
      \end{equation}
      Plugging this result in \eqref{eq:9_propo:super_lyap_beta_ula_0}, we get setting $A_{\gamma,\beta}(x) = \baretabeta(2^{-1} \Ltt(2+ \Ltt)\norm[2]{x} +  d+\baretabeta)$ that
      \begin{align}
  \label{eq:16_propo:super_lyap_beta_ula}
        &\parentheseDeux{\Qula_{\gamma}\Wbeta_{\baretabeta}(x)    - \rme^{-16^{-1}\gamma \mttbeta \baretabeta (1+\norm[2]{x})^{1/2}/(1+\norm[\beta]{x})} \Wbeta_{\baretabeta}(x)}\Big/ \Wbeta_{\baretabeta}(x)    \\
        &\qquad \qquad \qquad \qquad \leq
\rme^{\gamma A_{\gamma,\beta}(x)} - \rme^{-16^{-1}\gamma \mttbeta \baretabeta (1+\norm[2]{x})^{1/2}/(1+\norm[\beta]{x})}  \\
    \label{eq:16_propo:super_lyap_beta_ula_final}
&  \qquad \qquad \qquad \qquad \leq  \gamma \{A_{\gamma,\beta}(x)+16^{-1} \mttbeta \baretabeta (1+\norm[2]{x})^{1/2}/(1+\norm[\beta]{x})\} \rme^{\gamma A_{\gamma,\beta}(x)}   \eqsp,
\end{align}
where we have used that $\rme^{a_1} - \rme^{a_2} \leq \abs{a_1-a_2} \rme^{a_1 \vee a_2}$ for $a_1,a_2 \in \rset$. Combining this inequality with \eqref{eq:11_propo:super_lyap_beta_ula} completes the proof. 
\end{proof}

\begin{proposition}
  \label{v_beta:propo:lyap_mala_total}
  Assume \Cref{ass:regularity-U} and \Cref{ass:curvature_U_alpha}$(\beta)$ for $\beta \in\coint{0,1}$.  Then, there exist $\Gamma_{1/2,\beta} \geq  \Gamma_{\beta} > 0$ (given in \eqref{v_beta:eq:def_Gamma_1_2}-\eqref{v_beta:eq:def_Gamma} in the proof) such that for any $\bgamma \in \ocint{0,\Gamma_{\beta}}$, $\gamma \in\ocint{0,\bgamma}$ and $x \in\rset^d$,
  \begin{equation}
     R_{\gamma}\Wbeta_{\baretabeta}(x) \leq (1-\tildembeta \gamma)\Wbeta_{\baretabeta}(x)+ \bdriftmalabeta \gamma \1_{\ball{0}{\raymaladriftbeta}}(x) \eqsp,
  \end{equation}
  where $V_{\bareta}$ is defined by \eqref{eq:def_V_eta}, $R_{\gamma}$ is the Markov kernel of MALA defined by \eqref{eq:def-kernel-MALA},  $\baretabeta=\mttbeta/2^5$, $    \tildembeta  = \baretabeta \mttbeta \raymaladriftbeta^{1-\beta} /2^7$,
  \begin{align}
    \label{v_beta:eq:def_const_drift_1}
    \raymaladriftbeta  &= \max(1, 2 \Kttbeta ,  \rayulabeta, 2 \tKttbeta,\bKttbeta,[2^7 \tilde{b}_{1/2,\beta}/(\baretabeta\mttbeta)]^{1/(1-\beta)})  \eqsp, \\
        \label{v_beta:eq:def_const_drift_1_2}
    \tilde{b}_{1/2,\beta} &= C_{2,\Gamma_{1/2,\beta},\beta}\, d + \sup_{u \geq 1} \{u \rme^{-u/2^7}\} \eqsp, \\ 
    \label{v_beta:eq:def_b_drift}
    \bdriftmalabeta & = \bdriftulabeta + \gamma \baretabeta \mttbeta \raymaladriftbetaD^{1-\beta} \rme^{\baretabeta(1+\raymaladriftbetaD^2)^{1/2}}/2^7 + C_{1,\bgamma}  \bgamma^{1/2} \defEns{d + \sqrt{3}d^2 + \raymaladriftbeta^2}\eqsp,
  \end{align}
$\rayulabeta,\bdriftulabeta$ are defined in \Cref{propo:super_lyap_beta_ula}  and $C_{1,\bgamma}$ and $C_{2,\Gamma_{1/2,\beta},\beta}$ in \eqref{eq:def_C_1_bgamma} and \eqref{eq:def_C_2_bgamma_beta} respectively. 
\end{proposition}

\begin{proof}
  By \eqref{eq:diff-rula-rmala} and \Cref{propo:super_lyap_beta_ula},
for any $\bgamma \in \ocint{ \mttbeta/(2^5\Lttbeta^2)}$, $\gamma \in \ocintLigne{0,\bgamma}$ and $x \in\rset^d$,

  \begin{align}
  \Rkerg \Wbeta_{\baretabeta}(x) &
\leq \Qgam  \Wbeta_{\baretabeta}(x) + \Wbeta_{\baretabeta}(x)\int_{\rset^d} \{1-\min(1,\rme^{-\alphamala_{\gamma}(x,z)})\} \varphibf(z) \rmd z\\
 \label{v_beta:eq:propo:lyap_mala_total_1}
    & \leq \exp\defEns{-\frac{\gamma \mttbeta \baretabeta \norm[]{x}}{16(1+\norm[\beta]{x})}} \Wbeta_{\baretabeta}(x) + \gamma \bdriftulabeta \1_{\ball{0}{\rayulabeta}}(x) \\
    & \phantom{-----}+ \Wbeta_{\baretabeta}(x) \int_{\rset^d} \{1-\min(1,\rme^{-\alphamala_{\gamma}(x,z)})\} \varphibf(z) \rmd z \eqsp.
  \end{align}
  Let $\Upsilon_{\beta} \leq 1$ and 
  \begin{equation}
      \label{v_beta:eq:def_Gamma_1_2}
\Gamma_{1/2,\beta} = \min\parenthese{\Upsilon_{\beta}, (\mttbeta^3/(2^5 \Lttbeta^4)), (8d)^{-1}} \eqsp, \quad
\ray_{1/2} = \max\parenthese{1, 2 \Kttbeta ,  \rayulabeta, 2 \tKttbeta,\bKttbeta} \eqsp,
\end{equation}
where  $\tKttbeta$ is given in \Cref{lem:quadratic_behaviour_alpha} and $\bKttbeta$ in \Cref{lem:upper_bound_behaviour_alpha}.
Note that $\Gamma_{1/2,\beta} \leq \mttbeta/(2^5\Lttbeta^2)$ since $\mttbeta \leq \Lttbeta$ and $\ray_{1/2} \geq \sqrt{8\gamma d}$ for $\gamma \in \ocint{0,\Gamma_{1/2,\beta}}$.
Then,  by  \Cref{lem:bound_alpha_mala_2_alpha} and \Cref{lem_tail_chi2} with $c = 1/2^{5}$,  for any $x \in \rset^d$, $\norm{x} \geq \ray_{1/2}$, $\bgamma \in \ocint{0,\Gamma_{1/2,\beta}}$ and
 $\gamma \in \ocint{0,\bgamma}$,
  \begin{align}
    \Rkerg \Wbeta_{\baretabeta}(x) & \leq  \rme^{-\baretabeta \mttbeta \gamma \norm[1-\beta]{x}/2^5}\Wbeta_{\baretabeta}(x) +\Wbeta_{\baretabeta}(x) \defEns{C_{2,\Gamma_{1/2,\beta},\beta}\, d \gamma  + \exp(-\norm[2]{x}/(2^7\gamma))} \\
    & \leq  \rme^{-\baretabeta \mttbeta \gamma \norm[1-\beta]{x}/2^5}  \Wbeta_{\baretabeta}(x) + \Wbeta_{\baretabeta}(x) \gamma \tilde{b}_{1/2,\beta}  \eqsp,
  \end{align}
  where $ \tilde{b}_{1/2,\beta}$ is defined in \eqref{v_beta:eq:def_const_drift_1_2}.
  Note that $\raymaladriftbetaD = \max(\ray_{1/2},[2^7 \tilde{b}_{1/2,\beta}/(\baretabeta\mttbeta)]^{1/(1-\beta)})$ and   let
\begin{equation}
  \label{v_beta:eq:def_Gamma}
\Gamma_{\beta} = \min\parenthese{\Gamma_{1/2,\beta}, 2^{5}/\defEns{\mttbeta \baretabeta \raymaladriftbetaD^{1-\beta}}} \eqsp.
  \end{equation}
  Then, since for any $t \in \ccint{0,1}$, $\rme^{-t} \leq 1-t/2$, we get for any $x \in \rset^d$, $\norm{x} \geq \raymaladriftbetaD$, $\bgamma \in \ocint{0,\Gamma_{\beta}}$ and $\gamma \in \ocint{0,\bgamma}$,
  \begin{align}
    \nonumber
    \Rkerg \Wbeta_{\baretabeta}(x)& \leq  \rme^{-\baretabeta \mttbeta \gamma \raymaladriftbetaD^{1-\beta} /2^5} \Wbeta_{\baretabeta}(x) + \Wbeta_{\baretabeta}(x) \gamma \tilde{b}_{1/2,\beta} \\
    \label{v_beta:eq:drift_mala_totla_2}
    & \leq \parentheseDeux{1-\gamma\defEns{\baretabeta \mttbeta \raymaladriftbetaD^{1-\beta} /2^6 -\tilde{b}_{1/2,\beta}}}\Wbeta_{\baretabeta}(x) 
 \leq \defEns{1-\gamma \baretabeta \mttbeta \raymaladriftbetaD^{1-\beta} /2^7} \Wbeta_{\baretabeta}(x) \eqsp.
  \end{align}
  In addition, by \eqref{v_beta:eq:propo:lyap_mala_total_1} and  \Cref{lem:bound_alpha_mala_1}, using that for any $t \in \rset$, $1-\min(1,\rme^{-t}) \leq \abs{t}$, for any $x \in \rset^d$, $\norm{x} \leq \raymaladriftbetaD$, $\bgamma \in \ocint{0,\Gamma_{\beta}}$ and $\gamma \in \ocint{0,\bgamma}$,
  \begin{align}
    \Rkerg \Wbeta_{\baretabeta}(x) & \leq \Wbeta_{\baretabeta}(x) + \gamma \bdriftulabeta + C_{1,\bgamma} \gamma^{3/2} \int_{\rset^d} \{\norm[2]{z}+\norm[2]{x} + \norm[4]{z}\} \varphibf(z) \rmd z \\
    & \leq (1-\gamma \baretabeta \mttbeta \raymaladriftbetaD^{1-\beta} /2^7) \Wbeta_{\baretabeta}(x) + \gamma \bdriftulabeta + \gamma \baretabeta \mttbeta \raymaladriftbetaD^{1-\beta} \rme^{\baretabeta(1+\raymaladriftbetaD^2)^{1/2}}/2^7 \\
    &\phantom{-------------}+ C_{1,\bgamma} \gamma \bgamma^{1/2} \defEns{d + \sqrt{3}d^2 + \raymaladriftbetaD^2} \eqsp.
  \end{align}
  Combining this result and \eqref{v_beta:eq:drift_mala_totla_2} completes the proof.

\end{proof}

\subsection{Minorization condition}

\begin{lemma}
  \label{propo:lyap_mala_v2_just_lip}
  Assume \Cref{ass:regularity-U}  and let $\bgamma>0$. Then, for any $k \in\nsets$, $\gamma \in \ocint{0,\bgamma}$ and $x\in\rset^d$,
  \begin{equation}
    \int_{\rset^d} \norm[2]{y} \Qula_{\gamma}^{k}(x,\rmd y) \leq \rme^{k \gamma L_{\bgamma}}\norm[2]{x}+  2 \gamma d k \rme^{(k-1)\gamma L_{\bgamma}}  \eqsp,
  \end{equation}
  where $Q_{\gamma}$ is the Markov kernel of ULA defined in \eqref{eq:def_kernel_ula}, and 
  \begin{equation}
    \label{eq:def_L_bgamma}
    L_{\bgamma} = 2 \Ltt + \bgamma \Ltt^2 \eqsp.
  \end{equation}
\end{lemma}

\begin{proof}
  Let $\gamma\in\ocint{0,\bgamma}$ and $x\in\rset^d$. By \Cref{ass:regularity-U}, we have
  \begin{equation}
    \int_{\rset^d} \norm[2]{y} Q_{\gamma}(x,\rmd y) \leq
 \norm[2]{x}(1+2\gamma \Ltt + \gamma^2 \Ltt^2)+    2\gamma d \leq \norm[2]{x} \rme^{\gamma L_{\bgamma}} + 2 \gamma d  \eqsp,
  \end{equation}
  using that $1+t \leq \rme^t$ for $t \geq 0$. A straightforward induction completes the proof. 
\end{proof}

\begin{lemma}
  \label{v_beta_lem:diff_tv_MALA_ULA}
  Assume \Cref{ass:regularity-U} and let $\bgamma>0$. Then,  for any $x \in \rset^d$ and $\gamma \in \ocint{0,\bgamma}$, we have
\begin{align}
    \label{v_beta_eq:2:lem:diff_tv_MALA_ULA}
    \tvnorm{\updelta_x \Qgam^{\ceil{1/\gamma}} - \updelta_x \Rkerg^{\ceil{1/\gamma}}} &\leq C_{1,\bgamma} \gamma^{1/2} (d+\sqrt{3}d^2+\rme^{L_{\bgamma}}( \norm[2]{x}+2d)) \eqsp,
  \end{align}
  where $C_{1,\bgamma}$ is defined in \eqref{eq:def_C_1_bgamma} and $L_{\bgamma}$ in \eqref{eq:def_L_bgamma}.
\end{lemma}
\begin{proof}
  Let $x \in \rset^d$  and $\gamma \in \ocint{0,\bgamma}$.
 Consider the following decomposition
  \begin{equation}
    \updelta_x \Qgam^{\ceil{1/\gamma}} - \updelta_x \Rkerg^{\ceil{1/\gamma}} = \sum_{k=0}^{\ceil{1/\gamma}-1} \updelta_x \Qgam^k \{\Qgam - \Rkerg\} \Rkerg^{\ceil{1/\gamma}-k-1} \eqsp.
  \end{equation}
  Therefore using the triangle inequality, we obtain that
  \begin{equation}
            \label{v_beta_eq:4}
        \tvnorm{\updelta_x \Qgam^{\ceil{1/\gamma}} - \updelta_x \Rkerg^{\ceil{1/\gamma}}} \leq \sum_{k=0}^{\ceil{1/\gamma}-1} \tvnorm{ \updelta_x \Qgam^k \{\Rkerg - \Qgam\} \Rkerg^{\ceil{1/\gamma}-k-1}} \eqsp.
      \end{equation}
      We now bound each term in the sum. Let $k \in \{0,\ldots,\ceil{1/\gamma}-1\}$ and $f : \rset^d \to \rset$ be a bounded and measurable function. By   \eqref{eq:1:lem:diff_tv_MALA_ULA} in \Cref{lem:diff_tv_MALA_ULA}, we obtain that
      \[ \abs{ \updelta_x \{\Rkerg - \Qgam\} \Rkerg^{\ceil{1/\gamma}-k-1} f} \leq C_{1,\bgamma} \norm{f}_{\infty} \gamma^{3/2} \{d+\sqrt{3}d^2+\norm[2]{x}\} \]
      and therefore using \Cref{propo:lyap_mala_v2_just_lip}, we get
      \begin{equation}
\abs{       \updelta_x \Qgam^k   \{\Rkerg - \Qgam\} \Rkerg^{\ceil{1/\gamma}-k-1} f} \leq C_{1,\bgamma} \norm{f}_{\infty} \gamma^{3/2} \{d+\sqrt{3}d^2+\rme^{k L_{\bgamma} \gamma} \norm[2]{x} + 2\gamma k d \rme^{(k-1) L_{\bgamma} \gamma} \} \eqsp.
     \end{equation}
     Plugging this result in         \eqref{eq:4}, we obtain
     \begin{align}
       \tvnorm{\updelta_x \Qgam^{\ceil{1/\gamma}} - \updelta_x \Rkerg^{\ceil{1/\gamma}}} &\leq C_{1,\bgamma} \gamma^{3/2} \sum_{k=0}^{\ceil{1/\gamma}-1}   \{d+\sqrt{3}d^2+\rme^{k L_{\bgamma} \gamma} \norm[2]{x} + 2\gamma k d \rme^{(k-1) L_{\bgamma} \gamma} \}  \\
       &\leq C_{1,\bgamma} \gamma^{1/2} \{d+\sqrt{3}d^2+\rme^{ L_{\bgamma}} \norm[2]{x} + 2 d \rme^{ L_{\bgamma}} \}  \eqsp,
     \end{align}
     which concludes the proof.
\end{proof}

\begin{proposition}
  \label{v_beta_propo:small_set_mala}
  Assume \Cref{ass:regularity-U}. Then for any $\ray \geq 0$ there there exists $\hatGamma_{\ray} > 0$ (given in \eqref{eq:def_cgamma_ray_H1} in the proof), such that for any $\bgamma \in \ocintLigne{0,\hatGamma_{\ray}}$, $x,y \in \rset^d$, $\norm{x}\vee \norm{y} \leq \ray$, and $\gamma \in \ocint{0,\bgamma}$ we have
  \begin{equation}
\label{v_beta_eq:small_set_mala_propo}
    \tvnorm{\updelta_x \Rkerg^{\ceil{1/\gamma}} - \updelta_y \Rkerg^{\ceil{1/\gamma}}}   \leq 2 (1-\varepsilonula(\ray)/2) \eqsp,
  \end{equation}
  where $\varepsilonula(\ray)$ is defined in \eqref{eq:def_varepsi_ula}.
\end{proposition}

\begin{proof}
First note that for any $x,y \in \rset^d$, $\gamma >0$, by the triangle inequality, we obtain
  \begin{multline}
    \label{v_beta_eq:decomposiiton_small_set_mala_beta}
       \tvnorm{\updelta_x \Rkerg^{\ceil{1/\gamma}} - \updelta_y \Rkerg^{\ceil{1/\gamma}}} \leq        \tvnorm{\updelta_x \Rkerg^{\ceil{1/\gamma}} - \updelta_x \Qgam^{\ceil{1/\gamma}}}\\ + \tvnorm{\updelta_x \Qgam^{\ceil{1/\gamma}} - \updelta_y \Qgam^{\ceil{1/\gamma}}} + \tvnorm{\updelta_y \Rkerg^{\ceil{1/\gamma}} - \updelta_y \Qgam^{\ceil{1/\gamma}}} \eqsp.
     \end{multline}
     We now give some bounds for each term on the right hand side for any $x,y \in \rset^d$,  $\norm{x}\vee \norm{y} \leq \ray$ for a fixed $\ray \geq 0$ and $\gamma \leq  1/\Ltt$.
     By \Cref{propo:small_set_ula},  for any $x,y \in \rset^d$,  $\norm{x}\vee \norm{y} \leq \ray$ and $\gamma \leq  1/\Ltt$,
     \begin{equation}\label{v_beta_eq:bound_small_ULA_proof_small_MALA_beta_1}
       \tvnorm{\updelta_x \Qgam^{\ceil{1/\gamma}} - \updelta_y \Qgam^{\ceil{1/\gamma}}} \leq
       2(1-\varepsilonula(\ray)) \eqsp.
     \end{equation}
  In addition, by  \Cref{v_beta_lem:diff_tv_MALA_ULA},  for any $\bgamma >0$, $\gamma \in \ocint{0,\bgamma}$, and $z \in \rset^d$, $\norm{z} \leq \ray$,
  \begin{equation}
    \label{v_beta_eq:bound_small_ULA_proof_small_MALA_beta_2}
    \tvnorm{\updelta_z \Qgam^{\ceil{1/\gamma}} - \updelta_z \Rkerg^{\ceil{1/\gamma}}} \leq  C_{1,\bgamma}\gamma^{1/2}(d+\sqrt{3} d^2+\rme^{L_{\bgamma}}(\ray^2+2 d)) \eqsp.
  \end{equation}
Consider now for $\hat{\Upsilon}_{\beta} >0$, $\hat{\Gamma}_{1/2,\beta} = \hat{\Upsilon}_\beta \wedge \Ltt^{-1}$.  Combining \eqref{v_beta_eq:bound_small_ULA_proof_small_MALA_beta_1}-\eqref{v_beta_eq:bound_small_ULA_proof_small_MALA_beta_2} in \eqref{v_beta_eq:decomposiiton_small_set_mala_beta}, we obtain that for any $x,y\in \rset^d$, $\norm{x}\vee\norm{y} \leq \ray$, $\bgamma \in \ocintLigne{0,\hat{\Gamma}_{1/2,\beta}}$, $\gamma \in \ocint{0,\bgamma}$,
  \begin{equation}
    \norm{\updelta_x \Rkerg^{\ceil{1/\gamma}} - \updelta_y \Rkerg^{\ceil{1/\gamma}}} \leq 2(1-\varepsilonula(\ray)) + 2 C_{1,\hat{\Gamma}_{1/2,\beta} }\gamma^{1/2}(d+\sqrt{3} d^2+\rme^{L_{\hat{\Gamma}_{1/2,\beta} }}(\ray^2+2 d)) \eqsp.
  \end{equation}
  Therefore, we obtain that for any $x,y \in \rset^d$, $\norm{x} \vee \norm{y} \leq \ray$, $\gamma \in \ocint{0,\bgamma}$, \eqref{eq:small_set_mala_propo} holds  taking
  \begin{equation}
    \label{eq:def_cgamma_ray_H1}
    \hatGamma_{\ray} = \hat{\Gamma}_{1/2,\beta} \wedge \parentheseDeux{\frac{\varepsilon(\ray)}{2 C_{1,\hat{\Gamma}_{1/2,\beta} }(d+\sqrt{3} d^2+\rme^{L_{\hat{\Gamma}_{1/2,\beta} }}(\ray^2+2 d))}}^2 \eqsp.
  \end{equation}
\end{proof}

\subsection{Proof of \Cref{theo:V-geom_ergo_MALA_beta}}
\label{sec:proof:theo:V-geom_ergo_MALA_v_beta}

\Cref{v_beta:propo:lyap_mala_total} shows that there exist $\Gamma_{1/2,\beta} \geq  \Gamma_{\beta} > 0$ (given in \eqref{v_beta:eq:def_Gamma_1_2} and \eqref{v_beta:eq:def_Gamma}) such that for any $\bgamma \in \ocint{0,\Gamma_{\beta}}$, $\gamma \in\ocint{0,\bgamma}$ and $x \in\rset^d$,
  \begin{equation}
     R_{\gamma}\Wbeta_{\baretabeta}(x) \leq (1-\tildembeta \gamma)\Wbeta_{\baretabeta}(x)+ \bdriftmalabeta \gamma  \eqsp,
  \end{equation}
  where $\Wbeta_{\baretabeta}$ is defined by \eqref{eq:lyap_sub_quad_beta}, $R_{\gamma}$ is the Markov kernel of MALA defined by \eqref{eq:def-kernel-MALA},  $\baretabeta=\mttbeta/2^5$, $    \tildembeta, \bdriftmalabeta$ are specified in the statement of \Cref{v_beta:propo:lyap_mala_total}.
Using \cite[Lemma 14.1.10]{douc:moulines:priouret:soulier:2018}, we obtain $\pi(\Wbeta_{\baretabeta}) \leq A_{\bgamma,\beta} = \bdriftmalabeta/\tildembeta$. Therefore, we get
  \begin{equation}
    \label{v_beta:eq:def_bA_final}
    \pi(\Wbeta_{\baretabeta}) \leq \barA_{\beta} = \inf_{\bgamma \in \ocint{0,\Gamma_\beta}} A_{\bgamma,\beta} \eqsp.
  \end{equation}
We now show \eqref{v_beta:eq:5}.  Using $1-t \leq \rme^{-t}$ for $t \in \rset$ and setting $\lambda^{(\beta)} = \rme^{-\tildembeta}<1$,  an easy induction implies that for any $\gamma \in \ocint{0,\bgamma}$, $x \in \rset^d$,
\begin{equation}\label{eq:def-discrete-drift_2_v_beta}
  \RKer_\step^{\ceil{1/\gamma}}\Wbeta_{\baretabeta}(x) \leq \lambda^{(\beta)} \Wbeta_{\baretabeta}(x) +  \bdriftmalabeta(1+\bgamma) \eqsp.
\end{equation}
Set now
\begin{equation}
  \label{eq:def_M_gamma_v_beta_v_beta}
M_{\bgamma}^{(\beta)}  = \parenthese{\frac{4\bdriftmalabeta (1+\bgamma)}{1-\lambda^{(\beta)}}} \vee 1 \eqsp, \quad \ray_{\bgamma}^{(\beta)} = (\log(M_{\bgamma}^{(\beta)})/\baretabeta)^{1/2} \eqsp.
\end{equation}
Note that $\ball{0}{\ray_{\bgamma}^{(\beta)}}  = \{\Wbeta_{\baretabeta} \leq M_{\bgamma}^{(\beta)}\}$, $\bgamma \mapsto \bdriftmalabeta (1+\bgamma)$ and $\bgamma \mapsto \ray_{\bgamma}^{(\beta)}$ are increasing on $\rset_+$. Then, $\hatGamma_{\ray_{\bgamma}^{(\beta)},\beta} \geq \hatGamma_{\ray_{\Gamma_{\beta}}^{(\beta)},\beta}$ for $\bgamma \leq \Gamma_{\beta}$ where $\hatGamma_{\ray}$ is defined in \eqref{eq:def_cgamma_ray_H1}, and \Cref{propo:small_set_mala}  implies setting
\begin{equation}
  \label{eq:def_bGamma_final_v_beta}
  \bGamma_{\beta} = \Gamma_{\beta} \wedge \hatGamma_{\ray_{\Gamma_{\beta}}^{(\beta)}} \eqsp,
\end{equation}
that for any $\bgamma \in \ocint{0,\bGamma_{\beta}}$,  any $x,y\in \{\Wbeta_{\baretabeta} \leq M_{\bgamma}^{(\beta)}\}$, and $\gamma \in \ocint{0,\bgamma}$,
$\tvnorm{\updelta_{x} \Qmala^{\ceil{1/\gamma}} - \updelta_{y} \Qmala^{\ceil{1/\gamma}}} \leq 2(1-\varepsilon(\ray_{\bgamma}^{(\beta)}))$. As a result,
\cite[Theorem~19.4.1]{douc:moulines:priouret:soulier:2018} applied to $\Qmala^{\ceil{1/\gamma}}$ shows  that for any $x \in\rset^d$, $n \in \nset$, $\gamma \in \ocint{0,\bgamma}$,
\begin{equation}
 \Vnorm[\Wbeta_{\baretabeta}]{\updelta_x \Qmala^{n \ceil{1/\gamma}}_{\gamma} - \pi} \leq C_{\bgamma,\beta} \{ \Wbeta_{\baretabeta}(x) + \pi(\Wbeta_{\baretabeta}) \} \rho^n_{\bgamma,\beta} \eqsp,
  \end{equation}
  where 
  \begin{equation}
    \label{eq:cst_bornes_mala_v_beta}
      \begin{aligned}
        &\log \rho_{\bgamma,\beta}  = \frac{\log(1-2^{-1}\varepsilon(\ray_{\bgamma}^{(\beta)})) \log\bar\lambda^{(\beta)}} { \bigl(\log(1-2^{-1}\varepsilon(\ray_{\bgamma}^{(\beta)})) +
          \log\bar\lambda^{(\beta)}-\log {\barbdriftmalabeta}\bigr) } \eqsp ,\\
        &\bar\lambda^{(\beta)} = \lambda^{(\beta)} +(1-\lambda^{(\beta)})/2 \eqsp, \qquad 
 {\barbdriftmalabeta} = \lambda^{(\beta)} \bdriftmalabeta + M_{\bgamma}^{(\beta)}
        \eqsp,   \\
        &C_{\bgamma,\beta}  = \rho_{\bgamma,\beta}^{-1}\{\lambda^{(\beta)}+1\}\{1+{\barbdriftmalabeta}/[(1-2^{-1}\varepsilon(\ray_{\bgamma}^{(\beta)})(1-\bar\lambda^{(\beta)})]\} \eqsp .
  \end{aligned}
\end{equation}


\bibliographystyle{plain}
\bibliography{../bibliography/bibliography}

\appendix

\section{Technical results}

  \begin{lemma}
    \label{lem:equiv_str_conv_inf_hessian}
    Assume \Cref{ass:regularity-U} and \Cref{ass:curvature_U} hold. Then $U$ satisfies \eqref{eq:strong_convex_infi} with $\mtt' \leftarrow \mtt/2$, for any $x,y \in\rset^d$, $\norm{x} \vee \norm{y} \geq \Ktt + 8 \Ktt \Ltt/\mtt $. 
  \end{lemma}
  \begin{proof}
    For $x,y \in\rset^d$, $\norm{x} \vee \norm{y} \geq \Ktt + 8 \Ktt \Ltt/\mtt $, it holds that
    \begin{equation}
      \label{eq:Taylor_diff_nabla_U}
\ps{\nabla U(x) - \nabla U(y)}{x-y} = \int_{0}^1 \ps{\nabla^2U(x_t)\{x-y\}}{x-y}^{\otimes 2} \rmd t\eqsp, \text{ with $x_t = tx +(1-t)y$} \eqsp.
\end{equation}
Define $\msi = \{ t \in\ccint{0,1} \, : \, x_t \in \ball{0}{\Ktt}\}$. If $\msi = \emptyset$, then
by  \eqref{eq:Taylor_diff_nabla_U} and \Cref{ass:curvature_U}, we get
\begin{equation}
  \label{eq:equiv_str_conv_inf_hessian_1}
  \ps{\nabla U(x) - \nabla U(y)}{x-y} \geq \mtt \norm{x-y}^2 \eqsp. 
\end{equation}
If $\msi \neq  \emptyset$, necessarily,
\begin{equation}
  \label{eq:equiv_str_conv_inf_hessian_2}
  \norm{x-y} \geq 8 \Ktt \Ltt/\mtt \eqsp.
\end{equation}
Indeed, if this would not be true, using the triangle inequality we would have that $\norm{x_t} \geq \norm{x} \vee\norm{y} - \norm{x-y} \geq \norm{x} \vee \norm{y} - 8 \Ktt \Ltt/\mtt \geq \Ktt$ which would give a contradiction.
Now since $\msi \neq  \emptyset$ and is bounded, define $t_1 = \inf \msi$ and $t_2 = \sup\msi$. Note by definition, we have by continuity that $x_{t_1}, x_{t_2} \in \cball{0}{\Ktt}$ and therefore $\norm{x_{t_1}- x_{t_2}} \leq 2 \Ktt$. On the other hand, by definition, we have $\norm{x_{t_1}-x_{t_2}} = (t_2-t_1)\norm{x-y}$, so $t_2 - t_1 \leq 2\Ktt/\norm{x-y} \leq 1/4$ since $\Ltt \geq \mtt$. This implies by \eqref{eq:Taylor_diff_nabla_U}, the condition that $\nabla U $ is Lipschitz and \eqref{eq:equiv_str_conv_inf_hessian_2} that
\begin{align}
  \ps{\nabla U(x) - \nabla U(y)}{x-y} &\geq \mtt (1-(t_2-t_1)) \norm{x-y}^2 + \int_{t_1}^{t_2}  \ps{\nabla^2U(x_t)\{x-y\}}{x-y}^{\otimes 2} \rmd t \\
                                      & \geq  \mtt (1-(t_2-t_1)) \norm{x-y}^2  + \ps{\nabla U(x_{t_2}) - \nabla U(x_{t_1})}{x-y} \\
    \label{eq:equiv_str_conv_inf_hessian_3}
                                      & \geq (3/4) \mtt \norm{x-y}^2  - \Ltt (t_2-t_1)\norm{x-y}^2 \\
  & \geq (3/4) \mtt \norm{x-y}^2  - 2 \Ltt \Ktt \norm{x-y} \geq (\mtt /2) \norm{x-y}^2 \eqsp.
\end{align}
Combining this inequality with \eqref{eq:equiv_str_conv_inf_hessian_1} completes the proof. 
  \end{proof}

\begin{lemma}
  \label{lem:quadratic_behaviour}
    Assume \Cref{ass:regularity-U} and \Cref{ass:curvature_U} hold. The function $U$
  satisfies for any $x \in \rset^d$, $$\ps{\nabla U(x)}{x} \geq  (\mtt/2) \norm[2]{x} -\tCtt \1_{\ball{0}{\tKtt}}(x)\eqsp,$$ with $$\tKtt = 2 \Ktt(1+\Ltt/\mtt)\, \text{ and  }\tCtt = \Ltt \tKtt^2\eqsp.$$ 
\end{lemma}
\begin{proof}
  Using  \Cref{ass:regularity-U} and \Cref{ass:curvature_U} and the additional conditions we consider here on $U$, we have for any $x \not \in \ball{0}{\Ktt}$,
  \begin{align}
    \ps{\nabla U(x)}{x}
    &= \int_{0}^{\Ktt/\norm{x}} \DD^2 U(t x ) [x^{\otimes 2}] \rmd t + \int_{\Ktt/\norm{x}} ^ 1 \DD^2 U(t x ) [x^{\otimes 2}] \rmd t + \ps{\nabla U(0)}{x}\\
    & \geq \mtt\norm[2]{x} \{1- \Ktt (1 +\Ltt/\mtt)   / \norm{x}   \}\eqsp,
  \end{align}
  which implies that $\ps{\nabla U(x)}{x} \geq (\mtt/2) \norm[2]{x}$ for $\norm{x} \geq 2 \Ktt(1+\Ltt/\mtt)$. The final statement is an easy consequence of
  \Cref{ass:regularity-U} and the Cauchy-Schwarz inequality. 
\end{proof}

\begin{lemma}
  \label{lem:quadratic_behaviour_alpha}
    Assume \Cref{ass:regularity-U} and \Cref{ass:curvature_U_alpha}$(\beta)$ hold, for $\beta \in \coint{0,1}$. The function  $U$
  satisfies for any $x \in \rset^d$, $$\ps{\nabla U(x)}{x} \geq  (\mttbeta/2) \norm[2]{x}/(1+\norm[\beta]{x}) -\tCttbeta \1_{\ball{0}{\tKttbeta}}(x)\eqsp,$$ with $$\tKttbeta = [4 \Kttbeta(1+\Ltt/\mttbeta)]\vee[4\Kttbeta(1+\Ltt/\mttbeta)]^{1/(1-\beta)} \, \text{ and } \tCttbeta = \Ltt \tKttbeta^2 \eqsp.$$
\end{lemma}
\begin{proof}
  Using \Cref{ass:regularity-U} and \Cref{ass:curvature_U_alpha}$(\beta)$, we have for any $x \not \in \ball{0}{\Kttbeta}$,
  \begin{align}
    \ps{\nabla U(x)}{x}
    &= \int_{0}^{\Kttbeta/\norm{x}} \DD^2 U(t x ) [x^{\otimes 2}] \rmd t + \int_{\Kttbeta/\norm{x}} ^ 1 \DD^2 U(t x ) [x^{\otimes 2}] \rmd t + \ps{\nabla U(0)}{x}\\
    & \geq \frac{\mttbeta\norm[2]{x}}{1+\norm[\beta]{x}} \{1- (1+\norm[\beta]{x})\Kttbeta (1 +\Ltt/\mttbeta)   / \norm{x}   \}\eqsp,
  \end{align}
  which implies that $\ps{\nabla U(x)}{x} \geq (\mttbeta/2) \norm[2]{x}/(1+\normLigne[\beta]{x})$ for $\norm{x} \geq \tKttbeta$ distinguishing the case $\norm{x} \leq 1$ and $\norm{x} \geq 1$. The final statement is an easy consequence of
  \Cref{ass:regularity-U} and the Cauchy-Schwarz inequality. 
\end{proof}

\begin{lemma}
  \label{lem:upper_bound_behaviour_alpha}
    Assume \Cref{ass:regularity-U} and \Cref{ass:curvature_U_alpha}$(\beta)$ hold, for $\beta \in \coint{0,1}$. There exists $\bKttbeta \geq 0$ such that for any $x \not \in \ball{0}{\bKttbeta}$,  $\norm{\nabla U(x)} \leq 2 \Lttbeta \norm{x}/(1+\norm[3\beta/4]{x})$ with $\bKttbeta = [2\Ltt \Kttbeta/\Lttbeta] \vee [2\Ltt \Kttbeta/\Lttbeta]^{1/(1-3\beta/4)}$.
\end{lemma}
\begin{proof}
  Using \Cref{ass:regularity-U} and \Cref{ass:curvature_U_alpha}$(\beta)$, we have for any $x \not \in \ball{0}{\Kttbeta}$, $x \neq 0$, setting $t_{\beta} = \Kttbeta/\norm{x}$ and distinguishing the case $\norm{x}\leq 1$ and $\norm{x} \geq 1$, 
  \begin{align}
    \norm{\nabla U(x)} & = \normEq{\nabla U(0) + \int_0^{t_{\beta}} \nabla^2 U(tx) x \rmd t + \int_{t_\beta}^1 \nabla^2 U(tx) x \rmd t} \\
    & \leq \Ltt \Kttbeta+\Lttbeta \norm{x}/(1+\norm[3\beta/4]{x})  \leq 2\Lttbeta \norm{x}/(1+\norm[3\beta/4]{x}) \eqsp,
  \end{align}
  which completes the proof. 
\end{proof}

\begin{lemma}
  \label{lem:bounde_pertubhessian}
 Assume \Cref{ass:regularity-U}.
  Then, for any $t \in \ccint{0,1}$,  $\gamma \in \ocint{0,1/(4\Ltt)}$ and $x,z \in \rset^d$, $\norm{z} \leq \norm{x}/(4\sqrt{2\gamma})$, it holds
  \begin{equation}
    \norm{x+t\{-\gamma \nabla U(x) + \sqrt{2\gamma}z \}} \geq \norm{x}/2 \eqsp.
  \end{equation}
\end{lemma}

\begin{proof}
Note that for any $x \in\rset^d$, $\nabla U(x) = \nabla U(0) + \int_{0}^1 \nabla U(sx) x \rmd s$ and therefore  \Cref{ass:regularity-U} implies that $\nabla U(x) \leq \Ltt \norm{x}$. Let $t \in \ccint{0,1}$,  $\gamma \in \ocint{0,1/(4\Ltt)}$ and $x,z \in \rset^d$, $\norm{z} \leq \norm{x}/(4\sqrt{2\gamma})$.
  Using the triangle inequality, we have since $t \in \ccint{0,1}$
  \begin{equation}
    \norm{x+t\{-\gamma \nabla U(x) + \sqrt{2\gamma}z \}} \geq (1-\gamma \Ltt ) \norm{x} -\sqrt{2\gamma} \norm{z} \eqsp.
  \end{equation}
  The conclusion then follows from $\gamma \leq 1/(4\Ltt)$ and $\norm{z} \leq \norm{x}/(4\sqrt{2\gamma})$.
\end{proof}


\end{document}